\documentclass[letterpaper,twocolumn,10pt]{article}
\usepackage{usenix2019_v3}

\usepackage{filecontents}
\usepackage[utf8]{inputenc}
\usepackage{booktabs} 
\usepackage{graphicx}
\usepackage{amsmath}
\usepackage{subfigure}
\usepackage{units}
\usepackage{blindtext}
\usepackage{xcolor,xspace}
\usepackage{amsmath}
\usepackage{paralist}
\usepackage{amsmath}
\usepackage{tikz}
\usepackage{booktabs}
\usepackage{pifont}
\usepackage{multirow}
\usepackage{hhline}
\usepackage{balance}
\usepackage{footnote}
\usepackage{logicproof}
\usepackage{amsthm}
\usetikzlibrary{calc}
\usetikzlibrary{decorations.pathreplacing}

\usepackage{arydshln}
\usepackage{mdframed}
\usepackage{pgf-umlsd}

\usepackage{seqsplit}

\usepackage[T1]{fontenc}
\usepackage{
  beramono,
  listings,
  textcomp
}

\usetikzlibrary{arrows,automata}
\usetikzlibrary{calc}
\usetikzlibrary{arrows,calc,shapes,decorations.pathreplacing}

\newcommand\eqsize{\scriptsize}

\usepackage{listings}
\usepackage[ruled]{algorithm2e}
\usepackage{algorithmic}

\LinesNumbered

\usepackage{url}

\usepackage[
    n,
    advantage,
    operators,
    sets,
    adversary,
    landau,
    probability,
    notions,
    logic,
    ff,
    mm,
    primitives,
    events,
    complexity,
    asymptotics,
    keys]{cryptocode}

\newcommand{\acron}{{{\sf\it VRASED}}\xspace}
\newcommand{\vrased}{{{\sf\it VRASED}}\xspace}

\newcommand{\dev}{{\ensuremath{\sf{\mathcal Prv}}}\xspace}
\newcommand{\vrf}{{\ensuremath{\sf{\mathcal Vrf}}}\xspace}

\newcommand{\RA}{{\ensuremath{\sf{RA}}}\xspace}

\newcommand{\chal}{{\ensuremath{\sf{\mathcal Chal}}}\xspace}

\newcommand{\smart}{{\small SMART}\xspace}
\newcommand{\hydra}{{\small HYDRA}\xspace}
\newcommand{\sel}{\texttt{seL4}\xspace}
\newcommand{\attkey}{\ensuremath{\mathcal K}\xspace}

\newcommand{\cmark}{\ding{51}}%

\mathchardef\mhyphen="2D

\newcommand{\hw}{\texttt{\small HW-Mod}\xspace}
\newcommand{\sw}{\texttt{\small SW-Att}\xspace}
\newcommand{\swtiny}{\texttt{\tiny SW-Att}\xspace}
\newcommand{\rom}{\texttt{ROM}\xspace}
\newcommand{\ram}{\texttt{RAM}\xspace}
\newcommand{\hmac}{HMAC\xspace}

\newcommand{\dmaaddr}{\ensuremath{DMA_{addr}}\xspace}
\newcommand{\dmaen}{\ensuremath{DMA_{en}}\xspace}

\newcommand\blfootnote[1]{%
  \begingroup
  \renewcommand\thefootnote{}\footnote{#1}%
  \addtocounter{footnote}{-1}%
  \endgroup
}

\newtheorem{definition}{Definition}
\newtheorem{lemma}{Lemma}
\newtheorem{theorem}{Theorem}

\begin{document}

\title{Formally Verified Hardware/Software Co-Design for Remote Attestation}

\author{
{\rm Ivan De Oliveira Nunes}\\
University of California, Irvine\\
ivanoliv@uci.edu
\and
{\rm Karim Eldefrawy}\\
SRI International\\
karim.eldefrawy@sri.com
\and
{\rm Norrathep Rattanavipanon}\\
University of California, Irvine\\
nrattana@uci.edu
\and
{\rm Michael Steiner}\\
Intel\\
michael.steiner@intel.com
\and
{\rm Gene Tsudik}\\
University of California, Irvine\\
gene.tsudik@uci.edu
} 

\maketitle

\begin{abstract}
Remote Attestation (\RA) is a distinct security service that allows a trusted verifier (\vrf) to
measure the software state of an untrusted remote prover (\dev).
If correctly implemented, \RA allows \vrf to remotely detect if \dev is in an illegal or 
compromised state. Although several \RA approaches have been explored (including hardware-based, 
software-based, and hybrid) and many concrete methods have been proposed,
comparatively little attention has been devoted to formal verification.
In particular, thus far, no \RA designs and no implementations have been formally verified with 
respect to claimed security properties.

In this work, we take the first step towards formal verification of \RA by designing and verifying 
an architecture called \acron: \underline{\bf V}erifiable \underline{\bf R}emote \underline{\bf A}ttestation 
for \underline{\bf S}imple \underline{\bf E}mbedded 
\underline{\bf D}evices. \acron instantiates a hybrid (HW/SW) \RA co-design aimed at low-end embedded systems, 
e.g., simple IoT devices. \acron provides a level of security comparable to HW-based approaches, while relying on 
SW to minimize additional HW costs. Since security properties must be jointly guaranteed by HW and SW, verification  is a challenging task, which has never been attempted before in the context of \RA.
We believe that \acron is the first formally verified \RA scheme.
To the best of our knowledge, it is also the first formal verification of a HW/SW implementation of any security service. To demonstrate \acron's practicality and low overhead, we instantiate and evaluate it on a commodity platform (TI MSP430).
\acron's publicly available implementation was deployed on the Basys3 FPGA.

%
\end{abstract}
\section{Introduction}\label{intro}
\blfootnote{\textbf{\texttt{To appear: USENIX Security 2019.\\ Title: VRASED: A Verified Hardware/Software Co-Design for Remote Attestation}}}	
The number and variety of special-purpose computing devices is increasing dramatically. This includes 
all kinds of embedded devices, cyber-physical systems (CPS) and Internet-of-Things (IoT) gadgets, that are utilized in various 
``smart'' settings, such as homes, offices, factories, automotive systems and public venues.
As society becomes increasingly accustomed to being surrounded by, and dependent on, such devices,  their security becomes extremely important.
For actuation-capable devices, malware can impact both security and safety, e.g., as demonstrated by Stuxnet~\cite{stuxnet}.
Whereas, for sensing devices, malware can undermine privacy by obtaining ambient information.
Furthermore, clever malware can turn vulnerable IoT devices into zombies that can become sources for DDoS attacks.
For example, in 2016, a multitude of compromised ``smart'' cameras and DVRs formed the Mirai Botnet~\cite{antonakakis2017understanding} 
which was used to mount a massive-scale DDoS attack (the largest in history).

Unfortunately, security is typically not a key priority for low-end device manufacturers, due to cost, size or power constraints.
It is thus unrealistic to expect such devices to have the means to prevent current and future malware attacks.
The next best thing is detection of malware presence. This typically requires some form of {\bf Remote Attestation (\RA)} -- a distinct security service for detecting malware on CPS, embedded and IoT devices.
\RA is especially applicable to low-end embedded devices that are
incapable of defending themselves against malware infection. This is in contrast to more powerful devices (both embedded and general-purpose) 
that can avail themselves of sophisticated anti-malware protection.
\RA involves verification of current internal state  (i.e., \ram and/or flash) of an 
untrusted remote hardware platform (prover or \dev) by a trusted entity (verifier or \vrf). If \vrf detects malware presence, \dev's software can 
be re-set or rolled back and out-of-band measures can be taken to prevent similar infections. In general, \RA can help \vrf establish a static or 
dynamic root of trust in \dev and can also be used to construct other security services, such as software updates~\cite{seshadri2006scuba} 
and secure deletion~\cite{perito2010secure}.
Hybrid \RA (implemented as a HW/SW co-design) is a particularly promising approach for low-end embedded devices.
It aims to provide the same security guarantees as (more expensive) hardware-based approaches, while minimizing modifications to the underlying hardware.

Even though numerous \RA techniques with different assumptions, security guarantees, and designs, have been proposed~\cite{seshadri2006scuba,perito2010secure,Viper2011,smart,trustlite,tytan,hydra,brasser2016remote,Sancus17,SAP,erasmus,smarm,ibrahim2017seed,Sancus17,ConAsiaCCS18},
a major missing aspect of \RA is the high-assurance and rigor derivable from utilizing (automated) formal verification to guarantee security of the design and implementation of \RA techniques.
Because all aforementioned architectures and their implementations are not systematically designed from abstract models, their soundness and security can not be formally argued.
In fact, our \RA verification efforts revealed that a previous hybrid \RA design -- \smart~\cite{smart} -- 
assumed that disabling interrupts is an atomic operation and hence opened the door to compromise of \dev's secret key in the window between the time of the invocation of disable interrupts functionality and the time when interrupts are actually disabled.
Another low/medium-end architecture -- Trustlite~\cite{trustlite} -- also does not achieve our formal definition of \RA soundness.
In particular, this architecture is vulnerable to self-relocating malware (See~\cite{carpent2018reconciling} for details).
Formal specification of \RA properties and their (automated) verification significantly increases our confidence that such subtle issues are not overlooked.

In this paper we take a ``verifiable-by-design'' approach and develop, from scratch, an architecture for \underline{\bf V}erifiable \underline{\bf R}emote \underline{\bf A}ttestation for \underline{\bf S}imple 
\underline{\bf E}mbedded \underline{\bf D}evices (\acron). \acron is the first formally specified and verified \RA architecture accompanied by a formally verified  
implementation. Verification is carried out for all trusted components, including hardware, software, and the composition of both, all the way up to end-to-end notions for \RA soundness and security.
The resulting verified implementation -- along with its computer proofs -- is publicly available \cite{public-code}.
Formally reasoning about, and verifying, \acron involves overcoming major challenges that have not been attempted in the context of \RA and, to the best of our knowledge, not attempted for any security service implemented as a HW/SW co-design.
These challenges include:

\vspace{1mm}

 \noindent\textbf{1 --} Formal definitions of: \texttt{(i)} end-to-end notions for \RA soundness and security; \texttt{(ii)} a realistic machine model for low-end embedded systems; and \texttt{(iii)} \acron's guarantees.
 These definitions must be made in single formal system that is powerful enough to provide a common ground for reasoning about their interplay.
 In particular, our end goal is to prove that the definitions for \RA soundness and security are implied by \acron's guarantees when applied to our machine model.
 Our formal system of choice is Linear Temporal Logic (LTL). A background on LTL and our reasons for choosing it are discussed in Section~\ref{preliminaries}.

 \vspace{1mm} 

 \noindent\textbf{2 --} Automatic end-to-end verification of complex systems such as \acron is challenging from the computability perspective, as the space of possible states is extremely large.
 To cope with this challenge, we take a ``divide-to-conquer'' approach.
 We start by dividing the end-to-end goal of \RA soundness and security into smaller sub-properties that are also defined in LTL.
 Each HW sub-module, responsible for enforcing a given sub-property, is specified as a Finite State Machine (FSM), and verified using a Model Checker. 
 \acron's SW relies on an F* verified implementation (see Section~\ref{fermat_sw}) which is also specified in LTL.
 This modular approach allows us to efficiently prove sub-properties enforced by individual building blocks in \acron.

 \vspace{1mm}

 \noindent\textbf{3 --} All proven sub-properties must be composed together in order to reason about \RA security and soundness of \acron as one whole system. 
To this end, we use a theorem prover to show (by using LTL equivalences) that the sub-properties that were proved for each of \acron's sub-modules, when composed,
 imply the end-to-end definitions of \RA soundness and security. 
 This modular approach enables efficient system-wide formal verification.

\subsection{The Scope of Low-End Devices}
This work focuses on low-end devices based on low-power single core microcontrollers with a few KBytes 
of program and data memory. A representative of this class of devices is the Texas Instrument's MSP430 
microcontroller (MCU) family~\cite{TI-MSP430}. It has a $16$-bit word size, 
resulting in $\approx64$ KBytes of addressable memory. SRAM is used as data memory and its size ranges between $4$ and 
$16$KBytes (depending on the specific MSP430 model), while the rest of the address space is used for program memory, e.g., 
\rom and Flash. MSP430 is a Von Neumann architecture processor with common data and code 
address spaces. It can perform multiple memory accesses within a single instruction; 
its instruction execution time varies from $1$ to $6$ clock cycles, and instruction length varies from $16$ to $48$ bits. 
MSP430 was designed for low-power and low-cost. It is widely used in many application domains, e.g., automotive 
industry, utility meters, as well as consumer devices and computer peripherals. 
Our choice is also motivated by availability of a well-maintained open-source MSP430 hardware design from Open Cores \cite{openmsp430}.
Nevertheless, our machine model is applicable to other low-end MCUs in the same class as MSP430 (e.g., Atmel AVR ATMega).
\subsection{Organization}
Section \ref{preliminaries} provides relevant background on \RA and automated verification.
Section \ref{approach} contains the details of the \acron 
architecture and an overview of the verification approach. Section \ref{verif} contains the formal definitions of end-to-end \RA soundness and security and the formalization of the necessary sub-properties
along with the implementation of verified components to realize such sub-properties. Due to space limitation, the proofs for end-to-end soundness and security derived from the sub-properties are discussed in Appendix A.
Section \ref{alternatives} discusses alternative designs to guarantee the same required properties and their trade-offs with the standard design.
Section \ref{eval} presents experimental results demonstrating the minimal overhead of the formally verified and synthesized components.
Section~\ref{RW} discusses related work. Section~\ref{conclusion} concludes with a summary of our results.
End-to-end proofs of soundness and security, optional parts of the design, \acron's API, and discussion on
\acron's prototype can be found in Appendices A to D.

\section{Background}\label{preliminaries}
This section overviews \RA and provides some background on computer-aided verification.
\subsection{\RA for Low-end Devices}\label{sec:ra_bg}
As mentioned earlier, \RA is a security service that facilitates detection of malware presence on 
a remote device. Specifically, it allows a trusted verifier (\vrf) to
remotely measure the software state of an untrusted remote device (\dev).
As shown in Figure~\ref{fig:timeline}, \RA is typically obtained via a simple challenge-response protocol:
\begin{compactenum}
	\item \vrf sends an attestation request containing a challenge (\chal) to \dev. This request might also contain a token derived from a secret that allows \dev to authenticate \vrf.
	\item \dev receives the attestation request and computes an {\em authenticated integrity check} over its memory and \chal.
	The memory region might be either pre-defined, or 
		explicitly specified in the request.
        In the latter case, authentication of \vrf in step (1) is paramount to the overall security/privacy of \dev, as the request can specify arbitrary memory regions.
	\item \dev returns the result to \vrf. 
	\item \vrf receives the result from \dev, and checks whether it corresponds to a valid memory state.
\end{compactenum}

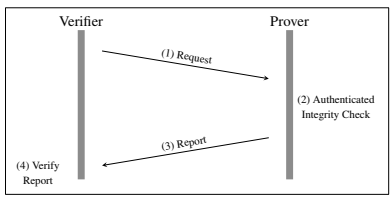
\begin{figure}[ht]
\centering
\resizebox{0.6\columnwidth}{!}{ 
	\fbox{
	\begin{tikzpicture}[node distance=1.5cm, >=stealth]
	\coordinate (BL)	at (0, 4.3);		\coordinate[left =6cm  of BL]	 (TL);
	\coordinate (Btvs)	at (0, 3.8);		\coordinate[left =6cm  of Btvs] (Ttvs);
	\coordinate (Btcs)	at (0, 2.8);		\coordinate[left =6cm  of Btcs] (Ttcs);
	\coordinate (Btce)	at (0, 1.3);		\coordinate[left =6cm  of Btce] (Ttce);
	\coordinate (Btvr)	at (0, 0.3);		\coordinate[left =6cm  of Btvr] (Ttvr);
	\coordinate (BR)	at (0, 0);		\coordinate[left =6cm  of BR]	 (TR);

	\node[above] at (BL) {\large Prover};
	\node[above] at (TL) {\large Verifier};
	\coordinate (Chksum) at ($(Btcs)!0.5!(Btce)$);
	\node [right = .1cm, align=center] at (Chksum) {\small (2) Authenticated \\ \small Integrity Check};
	\coordinate (Verify)	at ($(Ttvr)!0.5!(TR)$);
	\node [left =.5cm, align=center] at (Verify) {\small (4) Verify \\ \small Report};

	\draw[line width = .2cm, color=gray!90]	(BL) -- (BR);
	\draw[line width = .2cm, color=gray!90] (TL) -- (TR);
	\coordinate (ReqStart) at ($(Ttvs)!0.1!(Btcs)$);
	\coordinate (ReqEnd) at ($(Ttvs)!0.9!(Btcs)$);
	\draw[thick, ->] (ReqStart) -- (ReqEnd) node [above=-.05cm, midway, sloped] {\small(1) Request};
	\coordinate (RepStart) at ($(Btce)!0.1!(Ttvr)$);
	\coordinate (RepEnd) at ($(Btce)!0.9!(Ttvr)$);
	\draw[thick, ->] (RepStart) -- (RepEnd) node [above=-.05cm, midway, sloped] {\small(3) Report};
\end{tikzpicture}
	}
}
	\caption{\small Remote attestation (\RA) protocol}
	\label{fig:timeline}
\end{figure}

The {\em authenticated integrity check} can be realized as a Message Authentication Code (MAC) over \dev's memory.
However, computing a MAC requires \dev to have a unique secret key (denoted by \attkey) shared with \vrf. This \attkey 
must reside in secure storage,
where it is {\bf not} accessible to any software running on \dev, except for attestation code. Since most \RA threat models
assume a fully compromised software state on \dev, secure storage implies some level of hardware support.

Prior \RA approaches can be divided into three groups: software-based, hardware-based, and hybrid.
Software-based (or timing-based) \RA is the only viable approach for legacy devices with no hardware security
features. Without hardware support, it is (currently) impossible to guarantee that \attkey is not accessible by malware.
Therefore, security of software-based approaches~\cite{Seshadri:2005:PVC:1095809.1095812, Viper2011}
is attained by setting threshold communication delays between \vrf and \dev.
Thus, software-based \RA is unsuitable for multi-hop and jitter-prone communication, or settings where
a compromised \dev is aided (during attestation) by a more powerful
accomplice device. It also requires strong constraints and assumptions
on the hardware platform and attestation usage~\cite{KKWHAB2012,LiChGlPe2015}.
On the other extreme, hardware-based approaches require either i) \dev's attestation functionality to be housed entirely within 
dedicated hardware, e.g., Trusted Platform Modules (TPMs)~\cite{tpm}; or ii) modifications to the CPU  semantics or instruction sets to support the execution of trusted software, e.g., SGX~\cite{sgx} or TrustZone \cite{trustzone}.
Such hardware features are too expensive (in terms of physical area, energy consumption, and actual cost)
for low-end devices.%

While neither hardware- nor software-based approaches are well-suited for settings
where low-end devices communicate over the Internet (which is often the case in the IoT),
hybrid \RA (based on HW/SW co-design) is a more promising approach.
Hybrid \RA aims at providing the same security guarantees as hardware-based techniques with minimal hardware support.
\smart~\cite{smart} is the first hybrid \RA architecture targeting low-end MCUs.
In \smart, attestation's integrity check is implemented in software.
\smart's small hardware footprint guarantees that
the attestation code runs safely and that the attestation key is not leaked.
HYDRA~\cite{hydra} is a hybrid \RA scheme that relies on a secure boot hardware feature and on a secure micro-kernel.
Trustlite~\cite{trustlite} modifies Memory Protection Unit (MPU) and CPU exception engine hardware to implement \RA.
Tytan~\cite{tytan} is built on top of Trustlite, extending its capabilities for applications with real-time requirements.


Despite much progress, a major missing aspect in \RA research is high-assurance and rigor obtained by 
using formal methods to guarantee security of a concrete \RA design and its implementation
We believe that verifiability and formal security guarantees are particularly important for hybrid \RA designs 
aimed at low-end embedded and IoT devices, as their proliferation keeps growing. 
This serves as the main motivation for our efforts to develop the first formally verified \RA architecture.
\subsection{Formal Verification, Model Checking \& Linear Temporal Logic}
Computer-aided formal verification typically involves three basic steps.
First, the system of interest (e.g., hardware, software, communication protocol) must be described using a formal model, 
e.g., a Finite State Machine (FSM). Second, properties that the model should satisfy must be formally specified.
Third, the system model must be checked against formally specified properties to guarantee that the system retains 
such properties. This checking can be achieved via either Theorem Proving or Model Checking.
In this work, we use the latter and our motivation for picking it is clarified below.

In Model Checking, properties are specified as \textit{formulae} using Temporal Logic and system models are represented as  
FSMs. Hence, a system is represented by a triple $(S, S_0, T)$, where $S$ is a finite set of states,
$S_0 \subseteq S$ is the set of possible initial states, and $T \subseteq S \times S$ is the transition relation set, i.e., it describes 
the set of states that can be reached in a single step from each state. The use of Temporal Logic to specify properties allows representation 
of expected system behavior over time.

We apply the model checker NuSMV~\cite{cimatti2002nusmv}, which can be used to verify generic HW or SW models.
For digital hardware described at Register Transfer Level (RTL) -- which is the case in this work -- conversion
from Hardware Description Language (HDL) to NuSMV model specification is simple. Furthermore, it can be automated~\cite{irfan2016verilog2smv}.
This is because the standard RTL design already relies on describing hardware as an FSM.

In NuSMV, properties are specified in Linear Temporal Logic (LTL), which is particularly useful for verifying sequential 
systems. This is because it extends common logic statements with temporal clauses. In addition to propositional 
connectives, such as conjunction ($\land$), disjunction ($\lor$), negation ($\neg$), and implication ($\rightarrow$), 
LTL includes temporal connectives, thus enabling sequential reasoning. We are interested in the following temporal connectives:
\begin{compactitem}
 \item \textbf{X}$\phi$ -- ne\underline{X}t $\phi$: holds if $\phi$ is true at the next system state.
 \item \textbf{F}$\phi$ -- \underline{F}uture $\phi$: holds if there exists a future state where $\phi$ is true.
 \item \textbf{G}$\phi$ -- \underline{G}lobally $\phi$: holds if for all future states $\phi$ is true.
 \item $\phi$ \textbf{U} $\psi$ -- $\phi$ \underline{U}ntil $\psi$: holds if there is a future state where $\psi$ holds and
 $\phi$ holds for all states prior to that.
\end{compactitem}
This set of temporal connectives combined with propositional connectives (with their usual meanings) allows us to specify 
powerful rules. NuSMV works by checking LTL specifications against the system FSM for all reachable states in such FSM.
In particular, all \acron's desired security sub-properties are specified using LTL and verified by NuSMV.

\section{Overview of VRASED} \label{approach}
\acron is composed of a HW module (\hw) and a SW implementation  (\sw) of \dev's behavior according to the \RA protocol.
\hw enforces access control to \attkey in addition to secure and atomic execution of \sw (these properties are discussed in detail below).
\hw is designed with minimality in mind. 
The verified FSMs contain a minimal state space, which keeps hardware cost low.
\sw is responsible for computing an attestation report. 
As \acron's security properties are jointly enforced by \hw and \sw, both must be verified to ensure that the overall
design conforms to the system specification.

\subsection{Adversarial Capabilities \& Verification Axioms}
\label{subsec:axioms}
We consider an adversary, \adv, that can control the entire software state, code, and data of \dev.
\adv\ can modify any writable memory and read any memory that is not explicitly protected by access control rules, i.e.,
it can read anything (including secrets) that is not explicitly protected by \hw. It can also 
re-locate malware from one memory segment to another, in order to hide it from being detected. 
\adv\ may also have full control over all Direct Memory Access (DMA) controllers on \dev. 
DMA allows a hardware controller to directly access main memory (e.g., \ram, flash or \rom) without going through the CPU.

We focus on attestation functionality of \dev; verification of the entire MCU architecture is beyond the scope of this paper.
Therefore, we assume the MCU architecture strictly adheres to, and correctly implements, its specifications.
In particular, our verification approach relies on the following simple axioms:
\begin{compactitem}
 \item \textbf{A1 - Program Counter:} The program counter ($PC$) always contains the address of the instruction being 
 executed in a given cycle.
 \item \textbf{A2 - Memory Address:} Whenever memory is read or written, a data-address signal ($D_{addr}$) contains 
 the address of the corresponding memory location. For a read access, a data read-enable bit ($R_{en}$) must be set, and 
 for a write access, a data write-enable bit ($W_{en}$) must be set.
 \item \textbf{A3 - DMA:} Whenever a DMA controller attempts to access main system memory, a DMA-address signal 
 (\dmaaddr) reflects the address of the memory location being accessed and a DMA-enable bit (\dmaen) must be set. 
 DMA can not access memory when \dmaen is off (logical zero).
 \item \textbf{A4 - MCU reset:} At the end of a successful $reset$ routine, all registers (including $PC$) are set to zero 
 before resuming normal software execution flow. Resets are handled by the MCU in hardware; thus, reset handling routine 
 can not be modified.
 \item \textbf{A5 - Interrupts:} When interrupts happen, the corresponding $irq$ signal is set.
\end{compactitem}
\noindent \emph{\textbf{Remark:} Note that Axioms \textbf{A1} to \textbf{A5} are satisfied by the OpenMSP430 design.}

\sw uses the HACL*~\cite{hacl} \hmac-SHA256 function which is implemented and verified in F*\footnote{\url{https://www.fstar-lang.org/}}.
F* can be automatically translated to C and the proof of correctness for the translation is provided in~\cite{protzenko2017verified}.
However, even though efforts have been made to build formally verified C compilers (CompCert~\cite{compcert} is the most 
prominent example), there are currently no verified compilers targeting lower-end MCUs, such as MSP430. Hence, we assume that 
the standard compiler can be trusted to semantically preserve its expected behavior, especially with respect to the following:
\begin{compactitem}
 \item \textbf{A6 - Callee-Saves-Register:} Any register touched in a function is cleaned by default when the function 
 returns.
 \item \textbf{A7 - Semantic Preservation:} Functional correctness of the verified \hmac implementation 
 in C, when converted to assembly, is semantically preserved.
\end{compactitem}
\noindent \emph{\textbf{Remark:} Axioms \textbf{A6} and \textbf{A7} reflect the corresponding compiler specification (e.g., msp430-gcc).}

Physical hardware attacks are out of scope in this paper. Specifically, \adv\ can not modify code stored in \rom, induce 
hardware faults, or retrieve \dev secrets via physical presence side-channels. Protection against physical attacks is considered orthogonal 
and could be supported via standard tamper-resistance techniques~\cite{ravi2004tamper}.

\subsection{High-Level Properties of Secure Attestation}\label{high_prop}
We now describe, in high level, the sub-properties required for \RA.
In section~\ref{verif}, we formalize these sub-properties in LTL and provide single end-to-end definitions for \RA soundness and security.
Then we prove that \acron's design satisfies the aforementioned sub-properties and that the end-to-end definitions for soundness and security are implied by them.
The properties, shown in Figure~\ref{fig:properties}, fall into two groups: \emph{key protection} and \emph{safe execution}.\\

\begin{figure}
\centering
\includegraphics[width=0.7\columnwidth]{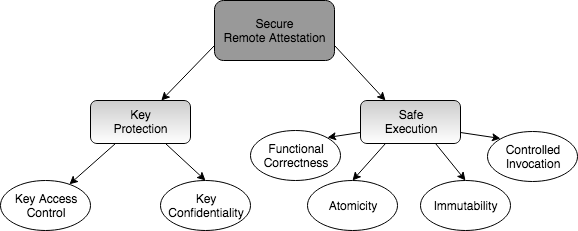}
\caption{Properties of secure \RA.}\label{fig:properties}
\end{figure}

\noindent \textbf{Key Protection:}

As mentioned earlier, \attkey must not be accessible by regular software running on \dev. 
%
To guarantee this, the following features must be correctly implemented:
\begin{compactitem}
 \item \textbf{P1- Access Control:} \attkey can only be accessed by \sw.
 \item \textbf{P2- No Leakage:} Neither \attkey (nor any function of \attkey
   other than the correctly computed HMAC) can remain in unprotected memory or registers after execution of \sw.
 \item \textbf{P3- Secure Reset:} Any memory tainted by \attkey and all registers (including PC) must be erased (or be inaccessible to regular 
 	software) after MCU reset. Since a reset might be triggered during \sw execution, lack of this property could result in leakage of 
	privileged information about the system state or \attkey.
	Erasure of registers as part of the reset ensures that no state from a previous execution persists. Therefore, the system 
	must return to the default initialization state.\\
\end{compactitem}

\noindent \textbf{Safe Execution:}

Safe execution ensures that \attkey is properly and securely used by 
\sw for its intended purpose in the \RA protocol. Safe execution can be divided into four sub-properties:

\begin{figure}
\centering
\resizebox{0.5\columnwidth}{!}{%
	\begin{tikzpicture}[node distance=1.5cm, >=stealth]
\tikzset{
  msp430/.style={
    draw,
    rectangle,
    minimum height=5cm,
    minimum width=3cm,
    align=center
  },
  membb/.style={
    draw,
    rectangle,
    minimum height=9cm,
    minimum width=1.5cm,
    align=center
  },
  mem/.style={
    draw,
    rectangle,
    minimum height=.5cm,
    minimum width=1.5cm,
    align=center
  },
  fermat/.style={
    draw,
    rectangle,
    minimum height=3cm,
    minimum width=1.5cm,
    align=center
  },
  arr/.style={->,>=latex'},
  three sided/.style={
        draw=none,
        append after command={
            [shorten <= -0.5\pgflinewidth]
            ([shift={( 0.5\pgflinewidth,-0.5\pgflinewidth)}]\tikzlastnode.north west)
        edge([shift={( 0.5\pgflinewidth,+0.5\pgflinewidth)}]\tikzlastnode.south west)            
            ([shift={( 0.5\pgflinewidth,+0.5\pgflinewidth)}]\tikzlastnode.south west)
        edge([shift={(-1.0\pgflinewidth,+0.5\pgflinewidth)}]\tikzlastnode.south east)       
            ([shift={( 0.5\pgflinewidth,+0.5\pgflinewidth)}]\tikzlastnode.south east)
        edge([shift={(-1.0\pgflinewidth,+0.5\pgflinewidth)}]\tikzlastnode.north east)
        }
    }
}
	\node[msp430] (msp) at (0,0) {\large MCU CORE};
        \node[membb] (bb) at (3.5,-2) {\shortstack{MEM.\\\ \\\ \\\ BACK-\\BONE}}; 
        \node[fermat, below = .2cm of bb.south] (fermat) {\hw};
        \node[mem, minimum height=2cm] (srom) at (6, 1.5) {\sw};
        \node[mem, below = .2cm of srom] (krom) {\attkey};
        \node[mem, below = .2cm of krom, minimum height=1cm] (sstack) {\shortstack{\sw \\ STACK \\ (XS)}};
        \node[mem, three sided, below = 0cm of sstack.south, minimum height=3.8cm] (ram) {\shortstack{App. \\ Avail. \\ RAM}};
        \node[mem, below = .2cm of ram, minimum height=4.05cm] (flash) {App. \\ Code};

        \path (msp.east) +(0,.3) coordinate (msp_2);
        \path (msp.east) +(0,-.3) coordinate (msp_3);
        \draw[thick, <-] (msp_2) -- (msp_2 -| bb.west);
        \draw[thick, ->] (msp_3) -- (msp_3 -| bb.west);
        
        \path (msp.south) +(.5,0) coordinate (msp_0);
        \path (msp.south) +(-.5,0) coordinate (msp_1);
        \path (fermat.west) +(0,.5) coordinate (fermat_0);
        \path (fermat.west) +(0,-.5) coordinate (fermat_1);
        \draw[thick, ->] (msp_0) |- (fermat_0); 
        \node[inner sep=5pt,right, fill=white,draw] at ($(msp_0) + (-.75,-2)$) {\small\shortstack{$PC$,\\$irq$,\\$R_{en}$,\\$W_{en}$,\\$D_{addr}$,\\$\dmaen$,\\$\dmaaddr$}};
        \draw[thick, <-] (msp_1) |- (fermat_1) ; 
        \node[inner sep=6pt,right, fill=white, draw] at ($(fermat_1) + (-2,0)$) {$reset$};
        
        \path (krom.west) +(0,.1) coordinate (krom_0);
        \path (krom.west) +(0,-.1) coordinate (krom_1);
        \draw[thick, ->] (krom_0) -- (krom_0 -| bb.east);
        \draw[thick, <-] (krom_1) -- (krom_1 -| bb.east);
        
        \path (srom.west) +(0,.1) coordinate (srom_0);
        \path (srom.west) +(0,-.1) coordinate (srom_1);
        \draw[thick, ->] (srom_0) -- (srom_0 -| bb.east); 
        \draw[thick, <-] (srom_1) -- (srom_1 -| bb.east); 
        
        \path (sstack.west) +(0,.1) coordinate (sstack_0);
        \path (sstack.west) +(0,-.1) coordinate (sstack_1);
        \draw[thick, ->] (sstack_0) -- (sstack_0 -| bb.east); 
        \draw[thick, <-] (sstack_1) -- (sstack_1 -| bb.east); 
        
        \path (ram.west) +(0,.1) coordinate (ram_0);
        \path (ram.west) +(0,-.1) coordinate (ram_1);
        \draw[thick, ->] (ram_0) -- (ram_0 -| bb.east); 
        \draw[thick, <-] (ram_1) -- (ram_1 -| bb.east); 
        
        \path (flash.north west) +(0,-.2) coordinate (flash_0);
        \path (flash.north west) +(0,-.4) coordinate (flash_1);
        \draw[thick, ->] (flash_0) -- (flash_0 -| bb.east); 
        \draw[thick, <-] (flash_1) -- (flash_1-| bb.east); 
        
        \draw [decorate,decoration={brace,amplitude=10pt,mirror,raise=4pt},yshift=0pt] 
        (krom.south east) -- (srom.north east) node [black,midway,xshift=.9cm] {\footnotesize ROM};
        
        \draw [decorate,decoration={brace,amplitude=10pt,mirror,raise=4pt},yshift=0pt] 
        (ram.south east) -- (sstack.north east) node [black,midway,xshift=.9cm] {\footnotesize RAM};
        
        \draw [decorate,decoration={brace,amplitude=10pt,mirror,raise=4pt},yshift=0pt] 
        (flash.south east) -- (flash.north east) node [black,midway,xshift=1cm] {\footnotesize FLASH};
\end{tikzpicture}
	}
\caption{VRASED system architecture}\label{fig:arch}
\end{figure}
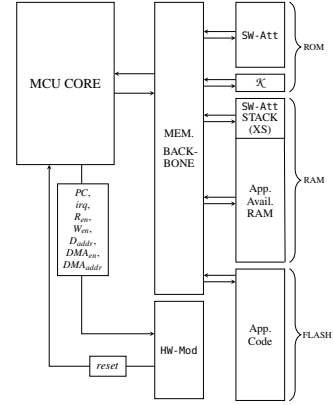

\begin{compactitem}
 \item \textbf{P4- Functional Correctness:} \sw must implement expected behavior of \dev's role in the \RA protocol.
 For instance, if \vrf expects a response containing an \hmac of memory in address range $[A,B]$, \sw implementation 
 should always reply accordingly. Moreover, \sw must always finish in finite time, regardless of input size and other parameters. 
 \item \textbf{P5- Immutability:} \sw executable must be immutable.
 Otherwise, malware residing in \dev could modify \sw, e.g., to always generate valid \RA measurements or to leak \attkey.
 \item \textbf{P6- Atomicity:} \sw execution can not be interrupted. The first reason for atomicity is to prevent leakage of 
 intermediate values in registers and \sw's data memory (including locations that could leak functions of \attkey) during \sw execution.
 This relates to \textbf{P2} above. The second reason is to prevent roving malware from relocating itself to escape being measured by \sw.
 \item \textbf{P7- Controlled Invocation:} \sw must always start from the first instruction and execute until the last instruction.
 Even though correct implementation of \sw is guaranteed by \textbf{P4}, isolated execution of chunks of a correctly implemented 
 code could lead to catastrophic results. Potential ROP attacks could be constructed using gadgets of \sw (which, based on \textbf{P1}, have 
 access to \attkey) to compute valid attestation results. 
\end{compactitem}

Beyond aforementioned core security properties, in some settings, \dev might need to authenticate \vrf's 
attestation requests in order to mitigate potential DoS attacks on \dev.
This functionality is also provided (and verified) as an optional feature in the design of \acron.
The differences between the standard design and the one with support for \vrf authentication are 
discussed in Appendix B.
\subsection{System Architecture}\label{sys-arch}
\acron architecture is depicted in Figure~\ref{fig:arch}. \acron is implemented by adding \hw to the MCU architecture, e.g., MSP430.
MCU memory layout is extended to include Read-Only Memory (\rom) that houses \sw code and  
\attkey used in the \hmac computation.
Because \attkey and \sw code are stored in \rom, we have guaranteed immutability, i.e., \textbf{P5}.
\acron also reserves a fixed part of the memory address space for \sw stack.
This amounts to $\approx3\%$ of the address space, as 
discussed in Section~\ref{eval} \footnote{A separate region in \ram is not strictly required. Alternatives and trade-offs 
are discussed in Section~\ref{alternatives}}. Access control to dedicated memory regions, as well as \sw atomic 
execution are enforced by \hw. The memory backbone is extended to support multiplexing of the new memory regions.
\hw takes $7$ input signals from the MCU core: $PC$, $irq$, $D_{addr}$, $R_{en}$, $W_{en}$, \dmaaddr and \dmaen. 
These inputs are used to determine a one-bit $reset$ signal output, that, when set to $1$, resets the MCU core immediately, i.e., 
before execution of the next instruction.
The $reset$ output is triggered when \hw detects any violation of security properties
\footnote{Resets due to \acron violations do not give malware advantages as malware can always trigger resets on the unmodified MCU by inducing software faults.}.

\subsection{Verification Approach}\label{verif_metho}
\begin{figure}[!hbtp]
\centering
\includegraphics[width=0.5\columnwidth]{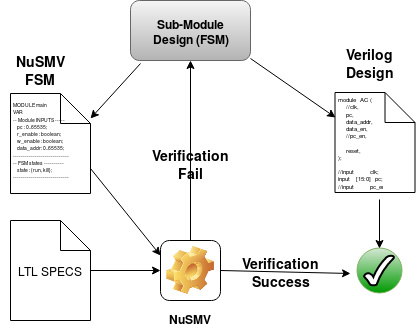}
\caption{\acron's submodule verification}\label{fig:module_verif}
\end{figure}
An overview of \hw verification is shown in Figures~\ref{fig:module_verif} and~\ref{fig:composition_verif}.
We start by formalizing \RA sub-properties discussed in this section using Linear Temporal Logic (LTL) 
to define invariants that must hold throughout the entire system execution.
\hw is implemented as a composition of sub-modules written in the Verilog hardware description language (HDL).
Each sub-module implements the hardware responsible for ensuring a given subset of the LTL specifications.
Each sub-module is described as an FSM in: (1) Verilog at Register Transfer Level (RTL); and
(2) the Model-Checking language SMV~\cite{cimatti2002nusmv}. We then use the NuSMV model checker to verify 
that the FSM complies with the LTL specifications. If verification fails, the sub-module is re-designed.

Once each sub-module is verified, they are combined into a single Verilog design.
The composition is converted to SMV using the automatic translation tool Verilog2SMV~\cite{irfan2016verilog2smv}.
The resulting SMV is simultaneously verified against all LTL specifications to prove that the final Verilog design for 
\hw complies with all secure \RA properties.

We clarify that the individual SMV sub-modules' design and verification steps are not strictly required in the verification pipeline.
This is because verifying SMV that is automatically translated from the composition of \hw would suffice.
Nevertheless, we design FSMs in SMV first so as to facilitate sub-modules' development and reasoning with an early additional check before going into their actual implementation and composition in Verilog.

\noindent \emph{\textbf{Remark:} Automatic conversion of the composition of \hw from Verilog to 
SMV rules out the possibility of human mistakes in representing Verilog FSMs as SMV.}\\

For the \sw part of \acron, we use the HMAC-SHA-256 from the HACL* library~\cite{hacl} to compute an authenticated
intregrity check of attested memory and \chal received from \vrf.
This function is formally verified with respect to memory safety, functional correctness, and cryptographic security.
However, key secrecy properties (such as clean-up of memory tainted by the key) are not formally verified in HACL* and thus must be ensured by \hw.%

As the last step, we prove that the conjunction of the LTL properties guaranteed by \hw and \sw implies soundness and security of the \RA architecture.
These are formally specified in Section~\ref{end-to-end-specs}.

\begin{figure}[!hbtp]
\centering
\includegraphics[width=0.5\columnwidth]{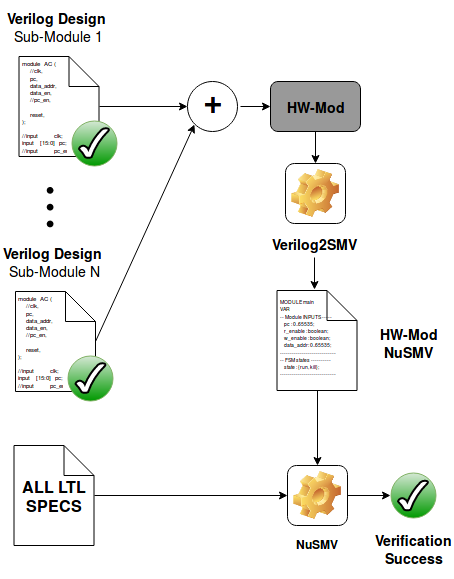}
\caption{Verification framework for the composition of sub-modules (HW-Mod).}\label{fig:composition_verif}
\end{figure}

\section{Verifying VRASED}\label{verif}
In this section we formalize \RA sub-properties.
For each sub-property, we represent it as a set of LTL specifications and construct an FSM that is verified to conform to such specifications.
Finally, the conjunction of these FSMs is implemented in Verilog HDL and translated to SMV using Verilog2SMV.
The generated SMV description for the conjunction is proved to simultaneously hold for all specifications.
We also define end-to-end soundness and security goals which are derived from the verified sub-properties (See Appendix A for the proof).

\subsection{Notation}
To facilitate generic LTL specifications that represent \acron's architecture (see Figure~\ref{fig:arch}) we use the following:
\begin{compactitem}
\item $AR_{min}$ and $AR_{max}$: first and last physical addresses of the memory region to be attested;
\item $CR_{min}$ and $CR_{max}$: physical addresses of first and last instructions of $\sw$ in \rom;
\item $K_{min}$ and $K_{max}$: first and last physical addresses of the \rom region where \attkey is stored;
\item $XS_{min}$ and $XS_{max}$: first and last physical addresses of the \ram region reserved for \sw computation;
\item $MAC_{addr}$: fixed address that stores the result of \sw computation (\hmac);
\item $MAC_{size}$: size of \hmac result;
\end{compactitem}
Table~\ref{tab:notation} uses the above definitions and summarizes the notation used in our LTL specifications throughout the rest of this paper.

To simplify specification of defined security properties, we use $[A,B]$ to denote a contiguous memory region 
between $A$ and $B$. Therefore, the following equivalence holds:
\begin{equation}
\eqsize
 C \in [A,B] \Leftrightarrow (C \leq B \land C \geq A)
\end{equation}
For example, expression $PC \in CR$ holds when the current value of $PC$ signal is within $CR_{min}$ and $CR_{max}$, 
meaning that the MCU is currently executing an instruction in CR, i.e, a \sw instruction.
This is because in the notation introduced above: $PC \in CR \Leftrightarrow PC \in [CR_{min},CR_{max}]  
\Leftrightarrow  (PC \leq CR_{max} \land PC \geq CR_{min})$.\\

\noindent\textbf{FSM Representation.}
As discussed in Section~\ref{approach}, \hw sub-modules are represented as FSMs that are verified to hold for LTL specifications.
These FSMs correspond to the Verilog hardware design of \hw sub-modules.
The FSMs are implemented as Mealy machines, where output changes at any time as a function of both 
the current state and current input values\footnote{This is in contrast with Moore machines where the output is 
defined solely based on the current state.}. Each FSM has as inputs a subset of the following signals and wires: 
$\{PC$, $irq$, $R_{en}$,$W_{en}$, $D_{addr}$, \dmaen, $\dmaaddr\}$.

Each FSM has only one output, $reset$, that indicates whether any security property was violated.
For the sake of presentation, we do not explicitly represent the value of the $reset$ output for each state.
Instead, we define the following implicit representation:
\begin{compactenum}
 \item $reset$ output is 1 whenever an FSM transitions to the $Reset$ state;
 \item $reset$ output remains 1 until a transition leaving the $Reset$ state is triggered;
 \item $reset$ output is 0 in all other states.
\end{compactenum}

\begin{table}
\centering
\caption{Notation summary}
\label{tab:notation}
\scriptsize{
\resizebox{\linewidth}{!}{%
\begin{tabular}{|l|p{9cm}|}
\hline
Notation    				&  Description  								\\ \hline \hline
$PC$					&  Current Program Counter value			\\ &\\
$R_{en}$					&  Signal that indicates if the MCU is reading from memory (1-bit)		\\ &\\
$W_{en}$					&  Signal that indicates if the MCU is writing to memory (1-bit)		\\ &\\
$D_{addr}$				& Address for an MCU memory access 		\\ &\\
\dmaen				&  Signal that indicates if DMA is currently enabled (1-bit)				\\ &\\
\dmaaddr				&  Memory address being accessed by DMA, if any 				\\ &\\
$CR$					&  (Code \rom) Memory region where \sw is stored: \\
						& $CR = [CR_{min}, CR_{max}]$   \\ & \\
$KR$					&  (\attkey \rom) Memory region where \attkey is stored: $KR = [K_{min}, K_{max}]$ \\ & \\
$XS$ 					&  (eXclusive Stack) secure \ram region reserved for \sw computations: $XS = [XS_{min},XS_{max}]$ \\ &\\
$MR$					& (MAC \ram) \ram region in which \sw computation result is written: $MR = [MAC_{addr},MAC_{addr} + MAC_{size}-1]$. The same region is also used to pass the attestation challenge as input to \sw \\ &\\
$AR$					& (Attested Region) Memory region to be attested. Can be fixed/predefined or specified in an authenticated request from \vrf: $AR = [AR_{min}, AR_{max}]$ \\ &\\
$reset$ 					& A 1-bit signal that reboots the MCU when set to logic $1$ \\ &\\
\hline \hline 
\textbf{A1, A2, ..., A7}		& Verification axioms (outlined in section \ref{subsec:axioms})\\ &\\
\textbf{P1, P2, ..., P7}		& Properties required for secure \RA (outlined in section \ref{high_prop})\\
\hline
\end{tabular}
}
}
\end{table}

\subsection{Formalizing \RA Soundness and Security}\label{end-to-end-specs}

We now define the notions of soundness and security.
Intuitively, \RA soundness corresponds to computing an integrity ensuring function over memory at time $t$.
Our integrity ensuring function is an \hmac computed on memory $AR$ with a one-time key derived from \attkey and \chal.
Since \sw computation is not instantaneous, \RA soundness must ensure that attested memory does not change during computation of the \hmac. This is the notion of temporal consistency in remote attestation~\cite{ConAsiaCCS18}.
In other words, the result of \sw call must reflect the entire state of the attested memory at the time when \sw is called. This notion is captured in LTL by Definition~\ref{ra_soundness}.

\begin{figure}[!ht]
\begin{mdframed}
\scriptsize
\begin{definition}{End-to-end definition for soundness of \RA computation}\label{ra_soundness}
\begin{align*}
\tiny
\begin{split}
 \textbf{G}:\ \{~& PC = CR_{min} \land AR = M  \land MR = \chal \ \land \ [(\neg reset) \ \textbf{U} \ (PC = CR_{max})] \rightarrow \\
 & \textbf{F}:\ [PC = CR_{max} \land MR = HMAC(KDF(\attkey,\chal), M)]\ \} \\
\end{split}
\end{align*}
where M is any AR value and KDF is a secure key derivation function.
\end{definition}
\end{mdframed}
\end{figure}

In Definition~\ref{ra_soundness}, $PC = CR_{min}$ captures the time when \sw is called (execution of its first instruction). $M$ and \chal are the values of $AR$ and $MR$.
From this pre-condition, Definition~\ref{ra_soundness} asserts that there is a time in the future when \sw computation finishes and, at that time, $MR$ stores the result of $HMAC(KDF(\attkey,\chal), M)$.
Note that, to satisfy Definition~\ref{ra_soundness}, \chal and $M$ in the resulting \hmac must correspond to the values in $AR$ and $MR$, respectively, when \sw was called.

\RA security is defined using the security game in Figure~\ref{fig:security_game}.
It models an adversary \adv\ (that is a probabilistic polynomial time, ppt, machine) that has full control of the software state of \dev (as the one described in Section~\ref{subsec:axioms}).
It can modify $AR$ at will and call \sw a polynomial number of times in the security parameter (\attkey and \chal bit-lengths). However, \adv\ can not modify \sw code, which is stored in immutable memory.
The game assumes that \adv\ does not have direct access to \attkey, and only learns \chal after it receives from \vrf as part of the attestation request.

\begin{figure}[!ht]
\begin{mdframed}
\scriptsize{
\begin{definition}\label{sec_def}~\\
\textbf{\ref{sec_def}.1 \RA Security Game (\RA-game):}

\texttt{Assumptions:}

- \swtiny is immutable, and \attkey is not known to \adv\\
- $l$ is the security parameter and $|\attkey| = |\chal| = |MR| = l$ \\
- $AR(t)$ denotes the content in $AR$ at time $t$\\
- \adv\ can modify $AR$ and $MR$ at will; however, it loses its ability to modify them while \swtiny is running 

\hrule
\vspace{2mm}

\texttt{\RA-game:}
	\begin{compactenum}
	 \item \texttt{Setup:} \adv\ is given oracle access to \swtiny.
     \item \texttt{Challenge:} A random challenge $\chal \leftarrow \$\{0,1\}^l$ is generated and given to \adv.
     \adv\ continues to have oracle access to \swtiny. 
     \item \texttt{Response:} Eventually, \adv\ responds with a pair $(M, \sigma)$, where $\sigma$ is either forged by \adv, or the result of calling \swtiny at some arbitrary time $t$.
     \item \adv\ wins if and only if $\sigma = HMAC(KDF(\attkey, \chal), M)$ and $M \neq AR(t)$.
	\end{compactenum}
	
\hrule
\vspace{2mm}
	
\textbf{\ref{sec_def}.2 \RA Security Definition:}\\
An \RA protocol is considered secure if there is no ppt \adv, polynomial in $l$, capable of winning the game defined in \ref{sec_def}.1 with
$Pr[\adv, \text{\RA-game}] > negl(l)$
\end{definition}
}
\end{mdframed}
\caption{\RA security definition for \acron} \label{fig:security_game}
\end{figure}

In the following sections, we define \sw functional correctness, LTL specifications~\ref{eq:AC_rule}-\ref{eq:dma-at} and formally verify that \acron's design guarantees such LTL specifications.
We define LTL specifications from the intuitive properties discussed in Section~\ref{high_prop} and depicted in Figure~\ref{fig:properties}.
In Appendix A we prove that the conjunction of such properties achieves soundness (Definition~\ref{ra_soundness}) and security (Definition \ref{sec_def}).
For the security proof, we first show that \acron guarantees that \adv\ can never learn \attkey with more than negligible probability, thus satisfying the assumption in the security game.
We then complete the proof of security via reduction, i.e., show that existence of an adversary that wins the game in Definition \ref{sec_def} implies 
the existence of an adversary that breaks the conjectured existential unforgeability of \hmac.

\noindent \emph{\textbf{Remark:} The rest of this section focuses on conveying the intuition behind the specification of LTL sub-properties.
Therefore, our references to the MCU machine model are via Axioms \textbf{A1 - A7} which were described in high level.
The interested reader can find an LTL machine model formalizing these notions in Appendix~A,
where we describe how such machine model is used construct computer proofs for Definitions~\ref{ra_soundness} and \ref{sec_def}.}

\subsection{\acron \sw} \label{fermat_sw}
To minimize required hardware features, hybrid \RA approaches implement integrity ensuring functions (e.g., \hmac) in software.
\acron's \sw implementation is built on top of HACL*'s \hmac implementation~\cite{hacl}.
HACL* code is verified to be functionally correct, memory safe and secret independent.
In addition, all memory is allocated on the stack making it predictable and deterministic.

\sw is simple, as depicted in Figure~\ref{fig:sw_att_code}.
It first derives a new unique context-specific key (${\sf\it key}$) from the master key (\attkey) by computing
an HMAC-based key derivation function, HKDF~\cite{krawczyk2010hmac}, on
\chal. 
This key derivation can be extended to incorporate attested memory boundaries
if \vrf specifies the range (see Appendix B).
Finally, it calls HACL*'s \hmac, using {\sf\it key } as the \hmac key.
$ATTEST\_DATA\_ADDR$ and $ATTEST\_SIZE$ specify the memory range to be attested ($AR$ in our notation).
We emphasize that \sw resides in \rom, which guarantees \textbf{P5} under the assumption 
of no hardware attacks.
Moreover, as discussed below, \hw enforces that no other software running on 
\dev can access memory allocated by \sw code, e.g., $key[64]$ buffer allocated in line 2 of Figure~\ref{fig:sw_att_code}.

\lstset{language=C,
	basicstyle={\scriptsize\ttfamily},
	showstringspaces=false,
	frame=single,
	xleftmargin=2em,
	framexleftmargin=3em,
	numbers=left, 
	numberstyle=\tiny,
	commentstyle={\tiny\itshape},
	keywordstyle={\tiny\bfseries},
	keywordstyle=\color{blue}\tiny\ttfamily,
	stringstyle=\color{red}\tiny\ttfamily,
        commentstyle=\color{black}\tiny\ttfamily,
        morecomment=[l][\color{magenta}]{\#},
        breaklines=true
}
\begin{figure}
\begin{lstlisting}[basicstyle=\tiny, numberstyle=\tiny]
void Hacl_HMAC_SHA2_256_hmac_entry() {
    uint8_t key[64] = {0};
    memcpy(key, (uint8_t*) KEY_ADDR, 64);
    hacl_hmac((uint8_t*) key, (uint8_t*) key, (uint32_t) 64, (uint8_t*) CHALL_ADDR, (uint32_t) 32);
    hacl_hmac((uint8_t*) MAC_ADDR, (uint8_t*) key, (uint32_t) 32, (uint8_t*) ATTEST_DATA_ADDR, (uint32_t) ATTEST_SIZE);
    return();
}
\end{lstlisting}
\caption{\sw C Implementation}%
\label{fig:sw_att_code}%
\end{figure}

HACL*'s verified \hmac is the core for guaranteeing \textbf{P4} (Functional Correctness) in \acron's design.
\sw functional correctness means that, as long as the memory regions storing values used in \sw computation ($CR$, $AR$, and $KR$) do not change during its computation,
the result of such computation is the correct HMAC. This guarantee can be formally expressed in LTL as in Definition~\ref{sw_fc}.
We note that since HACL*'s \hmac functional correctness is specified in F*, instead of LTL, we manually convert its guarantees to the LTL expressed by Definition~\ref{sw_fc}.
By this definition, the value in $MR$ does not need to remain the same, as it will eventually be overwritten by the result of \sw computation.

\begin{figure}[!ht]
\begin{mdframed}
\scriptsize
\begin{definition}{\swtiny functional correctness}\label{sw_fc}
 \begin{align*}
 \tiny
 \begin{split}
 \textbf{G}: \ \{~&
 PC = CR_{min} \land MR = \chal \ \land [(\neg reset \ \land \ \neg irq \ \land \ CR={\swtiny} \ \land \ KR = \attkey \ \land \ AR=M) \ \textbf{U} \ PC = CR_{max}] \\
 &\rightarrow \textbf{F}: \ [PC = CR_{max} \land MR = HMAC(KDF(\attkey,\chal), M)] \ \}
 \end{split}
 \end{align*}
 where M is any arbitrary value for AR.
\end{definition}
\end{mdframed}
\end{figure}

In addition, some HACL* properties, such as stack-based and deterministic memory allocation, are used in alternative 
designs of \acron to ensure \textbf{P2} -- see Section~\ref{alternatives}.

Functional correctness implies that the \hmac implementation conforms to its published standard specification 
on all possible inputs, retaining the specification's cryptographic security. It also implies that HMAC executes in finite time.
Secret independence ensures that there are no branches taken as a function of secrets, i.e., 
\attkey and {\sf\it key} in Figure~\ref{fig:sw_att_code}. This mitigates \attkey leakage via timing side-channel attacks.
Memory safety guarantees that implemented code is type safe, meaning that it never reads from, or writes to: invalid 
memory locations, out-of-bounds memory, or unallocated memory. This is particularly important for preventing ROP attacks, 
as long as \textbf{P7} (controlled invocation) is also preserved\footnote{Otherwise, even though the implementation is 
memory-safe and correct as a whole, chunks of a memory-safe code could still be used in ROP attacks.}.

Having all memory allocated on the stack allows us to either: (1) confine \sw execution to a fixed size protected memory region
inaccessible to regular software (including malware) running on \dev; or (2) ensure that \sw stack is erased before the end
of execution. Note that HACL* 
does not provide 
stack erasure, in order to improve performance. Therefore, \textbf{P2} does not follow from HACL* implementation.
However, erasure before \sw terminates must be guaranteed.
Recall that \acron targets low-end MCUs that might run 
applications on bare-metal and thus can not rely on any OS features.

As discussed above, even though HACL* implementation guarantees \textbf{P4} and storage in \rom guarantees \textbf{P5},
these must be combined with \textbf{P6} and \textbf{P7} to provide safe execution. \textbf{P6} and \textbf{P7} -- along with the key 
protection properties (\textbf{P1}, \textbf{P2}, and \textbf{P3}) -- are ensured by \hw, which are described next.

\subsection{Key Access Control (\hw)}
If malware manages to read \attkey from \rom, it can reply to \vrf with a forged result.
\hw access control (AC) sub-module enforces that \attkey can only be accessed by \sw (\textbf{P1}).\\


\subsubsection{\textbf{LTL Specification}}
The invariant for key access control (AC) is defined in LTL Specification~(\ref{eq:AC_rule}). It stipulates that system must transition to the
$Reset$ state whenever code from outside $CR$ tries to read from 
$D_{addr}$ within the key space.
\begin{equation}\label{eq:AC_rule}
\eqsize
\begin{split}
\text{\bf G}: \ \{
\neg (PC \in CR) \land R_{en} \land (D_{addr} \in KR) \rightarrow reset \ \}
\end{split}
\end{equation}
\subsubsection{\textbf{Verified Model}}
Figure~\ref{fig:AC_FSM} shows the FSM implemented by the AC sub-module which is verified to hold for LTL 
Specification~\ref{eq:AC_rule}. This FSM has two states: \textit{Run} and \textit{Reset}.  It outputs $reset=1$ 
when the AC sub-module transitions to state \textit{Reset}. This implies a hard-reset of the MCU.
Once the reset process completes, the system leaves the \textit{Reset} state.

\begin{figure}
\begin{center}
\noindent\resizebox{0.7\columnwidth}{!}{%
	\begin{tikzpicture}[->,>=stealth',auto,node distance=8.0cm,semithick]
		\tikzstyle{every state}=[minimum size=1.5cm]
		\tikzstyle{every node}=[font=\large]

		\node[state] 		(A)					{$Run$};
		\node[state]         (B) [right of=A,align=center]	{$Reset$};

		\path[->,every loop/.style={looseness=8}] 
			(A)
  				edge [loop above] node {$otherwise$} (A)
			(B)  
  				edge [loop above] node {$otherwise$} (B);
  		
\draw[->] (A.345) -- node[rotate=0,below, align=center,auto=right] {\small$\neg(PC\in CR)\,\land R_{en}\,\land\,(D_{addr}\in KR)$} (B.195);
\draw[<-] (A.15) -- node[rotate=0,above] {\small$PC=0$} (B.165);
	\end{tikzpicture}
}
\caption{Verified FSM for Key AC}
\label{fig:AC_FSM}
\end{center}
\end{figure}
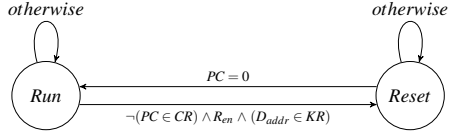

\subsection{Atomicity and Controlled Invocation (\hw)}
In addition to functional correctness, safe execution of attestation code requires immutability (\textbf{P5}), 
atomicity (\textbf{P6}), and controlled invocation (\textbf{P7}). \textbf{P5} is achieved directly by placing \sw in \rom.
Therefore, we only need to formalize invariants for the other two properties: atomicity and controlled execution.

\subsubsection{\textbf{LTL Specification}} 
To guarantee atomic execution and controlled invocation, LTL Specifications~(\ref{eq:at1}), (\ref{eq:at2}) and (\ref{eq:irq}) must hold:
\begin{equation}\label{eq:at1}
\eqsize
\begin{split}
 & \text{\bf G}: \ \{ 
  [\neg reset \land
  (PC \in CR) \land
  \neg (\text{\bf X}(PC) \in CR)] \rightarrow 
  [PC = CR_{max} \lor \text{\bf X}(reset)] \ \}
\end{split}
\end{equation}
\begin{equation}\label{eq:at2}
\eqsize
\begin{split}
 & \text{\bf G}: \ \{ 
 [\neg reset \land
 \neg (PC \in CR) \land
  (\text{\bf X}(PC) \in CR)] \rightarrow
 [\text{\bf X}(PC) = CR_{min} \lor \text{\bf X}(reset)] \ \}
\end{split}
\end{equation}
\begin{equation}\label{eq:irq}
\eqsize
\begin{split}
 & \text{\bf G}: \ \{
  irq ~\land~
  (PC \in CR)~ \rightarrow
 ~reset \ \}
\end{split}
\end{equation}

LTL Specification~(\ref{eq:at1}) enforces that the only way for \sw execution to terminate is through its last instruction: $PC = CR_{max}$.
This is specified by checking current and next $PC$ values using LTL ne\textbf{X}t operator.
In particular, if current $PC$ value is within \sw region, and next $PC$ value is out of \sw region, then either
current $PC$ value is the address of the last instruction in \sw ($CR_{max}$), or $reset$ is triggered in the next cycle.
Also, LTL Specification~(\ref{eq:at2}) enforces that the only way for $PC$ to enter \sw region is through the very first instruction: $CR_{min}$.
Together, these two invariants imply \textbf{P7}: it is impossible to jump into the middle of \sw, or to leave 
\sw before reaching the last instruction. 

\textbf{P6} is satisfied through LTL Specification~(\ref{eq:irq}).
Atomicity could be violated by interrupts. However, LTL Specification (\ref{eq:irq}) prevents an interrupt to happen while \sw is executing.
Therefore, if interrupts are not disabled by software running on \dev before calling \sw, any interrupt that might 
violate \sw atomicity will cause an MCU $reset$.

\subsubsection{\textbf{Verified Model}}
Figure~\ref{fig:atomicity_fsm} presents a verified model for atomicity and controlled invocation enforcement.
The FSM has five states.
Two basic states $notCR$ and $midCR$ represent moments when $PC$ points to an address: (1) outside $CR$, 
and (2) within $CR$, respectively, not including the first and last instructions of \sw.
Another two: $fstCR$ and $lstCR$ represent states when $PC$ points to the first and last instructions of \sw, respectively.
Note that the only possible path from $notCR$ to $midCR$ is through $fstCR$.
Similarly, the only path from $midCR$ to $notCR$ is through $lstCR$.
The FSM transitions to the $Reset$ state whenever:
(1) any sequence of values for $PC$ does not obey the aforementioned conditions; or
(2) $irq$ is logical 1 while exeucuting \sw.

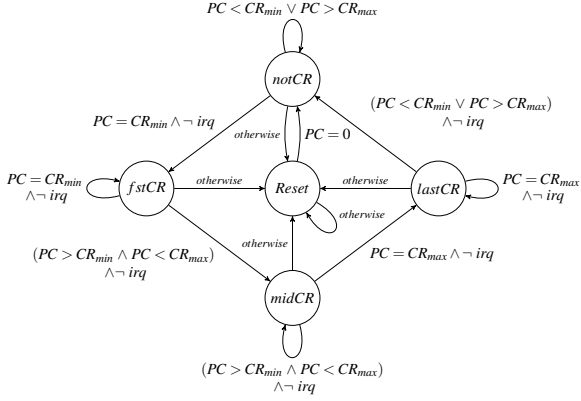
\begin{figure}
\begin{center}
\noindent\resizebox{0.9\columnwidth}{!}{%
	\begin{tikzpicture}[->,>=stealth',auto,node distance=4.0cm,semithick]
		\tikzstyle{every state}=[minimum size=1.5cm]
		\tikzstyle{every node}=[font=\large]

		\node[state] 		(A)					{$Reset$};
		\node[state]         (B) [above of=A,align=center, yshift=-1cm]	{$notCR$};
		\node[state]         (C) [left  of=A,align=center]	{$fstCR$};
		\node[state]         (E) [below of=A, yshift=1cm] {$midCR$};
		\node[state]         (D) [right of=A]	{$lastCR$};

		\path[->,every loop/.style={looseness=8}] (A) edge [bend right=10]  node [right] {$PC=0$} (B)
				edge [out=330,in=300,looseness=8] node[above right, yshift=.2cm] {\small$otherwise$} (A)
			(B)  edge [loop above] node {$PC<CR_{min}\,\lor\,PC>CR_{max}$} (C)
				edge node [above left] {$PC=CR_{min}\land \neg~irq$} (C)
				edge [bend right=10] node [left] {\small$otherwise$} (A)
			(C)  edge [loop left] node [above right, left] {\shortstack{$PC=CR_{min}$\\~$\land \neg~irq$}} (C)
				edge node [below left] {\shortstack{$(PC>CR_{min}\,\land\,PC<CR_{max})$\\~$\land \neg~irq$}} (E)
				edge node {\small$otherwise$} (A)
			(E)  edge [loop below] node {\shortstack{$(PC>CR_{min}\,\land\,PC<CR_{max})$\\~$\land \neg~irq$}} (C)
				edge node [below right] {$PC=CR_{max}\land \neg~irq$} (D)
				edge node [left] {\small$otherwise$} (A)
			(D)  edge [loop right] node [above left, right]  {\shortstack{$PC=CR_{max}$\\~$\land \neg~irq$}} (D)
				edge node [above right] {\shortstack{$(PC<CR_{min}\,\lor\,PC>CR_{max})$\\~$\land \neg~irq$}} (B)
				edge node [above]  {\small$otherwise$} (A);
	\end{tikzpicture}
}
\caption{Verified FSM for atomicity and controlled invocation.}
\label{fig:atomicity_fsm}
\end{center}
\end{figure}

\subsection{Key Confidentiality (\hw)}
To guarantee secrecy of \attkey and thus satisfy \textbf{P2}, \acron must enforce the following:
\begin{compactenum}
\item No leaks after attestation: any registers and memory accessible to applications must be erased
at the end of each attestation instance, i.e., before application execution resumes.
\item No leaks on reset: since a reset can be triggered during attestation execution, any registers and
memory accessible to regular applications must be erased upon reset.
\end{compactenum}
Per Axiom \textbf{A4}, all registers are zeroed out upon reset and at boot time.
Therefore, the only time when register clean-up is necessary is at the end of \sw.
Such clean-up is guaranteed by the Callee-Saves-Register convention: Axiom \textbf{A6}.

Nonetheless, the leakage problem remains because of \ram allocated by \sw.
Thus, we must guarantee that \attkey is not leaked through "dead" memory, which could be accessed by
application (possibly, malware) after \sw terminates. A simple and effective way of addressing this issue 
is by reserving a separate secure stack in \ram that is only accessible (i.e., readable and writable) by 
attestation code. All memory allocations by \sw must be done on this stack, and access control to the stack 
must be enforced by \hw. As discussed in Section~\ref{eval}, the size of this stack is constant -- $2.3$KBytes. 
This corresponds to $\approx 3\%$ of MSP430 16-bit address space.
We discuss \acron variants that do not require a reserved stack and trade-offs between them in Section~\ref{alternatives}.

\subsubsection{\textbf{LTL Specification}}
Recall that $XS$ denote a contiguous secure memory region reserved for exclusive access by \sw. 
LTL Specification for the secure stack sub-module is as follows:
\begin{equation}\label{eq:stackrule1}
\eqsize
\begin{split}
& \text{\bf G}: \ \{ 
\neg (PC \in CR) \land
 (R_{en} \lor W_{en}) \land
 (D_{addr} \in XS) \rightarrow 
 reset \ \}
\end{split}
\end{equation}
We also want to prevent attestation code from writing into application memory. Therefore, it is only allowed to write to the 
designated fixed region for the HMAC result ($MR$).
\begin{equation}\label{eq:stackrule2}
\eqsize
\begin{split}
& \text{\bf G}: \ \{ (PC \in CR) \land
 (W_{en}) \land
 \neg (D_{addr} \in XS) \land 
 \neg (D_{addr} \in MR) \rightarrow reset \ \}
\end{split}
\end{equation}
In summary, invariants (\ref{eq:stackrule1}) and (\ref{eq:stackrule2}) enforce that only attestation code can read from/write 
to the secure reserved stack and that attestation code can only write to regular memory within the space reserved 
for the HMAC result. If any of these conditions is violated, the system resets.

\subsubsection{\textbf{Verified Model}}
Figure~\ref{fig:stack_FSM} shows the FSM verified to comply with invariants~(\ref{eq:stackrule1}) and~(\ref{eq:stackrule2}).
%
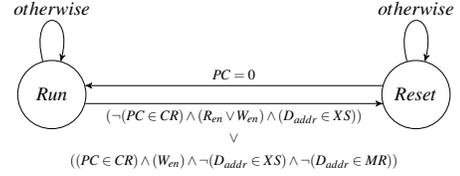
\begin{figure}
\begin{center}
\noindent\resizebox{0.7\columnwidth}{!}{%
	\begin{tikzpicture}[->,>=stealth',auto,node distance=8.0cm,semithick]
		\tikzstyle{every state}=[minimum size=1.5cm]
		\tikzstyle{every node}=[font=\large]

		\node[state] 		(A)					{$Run$};
		\node[state]         (B) [right of=A,align=center]	{$Reset$};

		\path[->,every loop/.style={looseness=8}] 
			(A)
  				edge [loop above] node {$otherwise$} (A)
			(B)  
  				edge [loop above] node {$otherwise$} (B);
  		
\draw[->] (A.345) -- node[rotate=0,below, align=center,auto=right] {\small$(\neg (PC \in CR) \land (R_{en} \lor W_{en}) \land (D_{addr} \in XS))
$\\\small$\lor$\\
\small$ ((PC \in CR) \land (W_{en}) \land \neg (D_{addr} \in XS) \land \neg (D_{addr} \in MR))$} (B.195);
\draw[<-] (A.15) -- node[rotate=0,above] {\small$PC=0$} (B.165);
	\end{tikzpicture}
}
\caption{Verified FSM for Key Confidentiality}
\label{fig:stack_FSM}
\end{center}
\end{figure}

\subsection{DMA Support}\label{sec:dma-sup}
So far, we presented a formalization of \hw sub-modules under the assumption that DMA is either 
not present or disabled on \dev. However, when present, a DMA controller can access 
arbitrary memory regions. Such memory access is performed concurrently in the memory backbone 
and without MCU intervention, while the MCU executes regular instructions.

DMA data transfer is performed using dedicated memory buses, e.g., \dmaen and \dmaaddr.
Hence, regular memory access control (based on monitoring $D_{addr}$) does not apply to memory access 
by DMA controller. Thus, if DMA controller is compromised, it may lead to violation of {\bf P1} and {\bf P2} by 
directly reading \attkey and values in the attestation stack, respectively. In addition, it can assist \dev-resident 
malware to escape detection by either copying it out of the measurement range or deleting it, 
which results in a violation of {\bf P6}.

\subsubsection{\textbf{LTL Specification}}
We introduce three additional LTL Specifications to protect against aforementioned attacks.
First, we enforce that DMA cannot access \attkey.
\begin{equation}\label{eq:DMA_AC_rule}
\eqsize
\begin{split}
& \text{\bf G}: \ \{ \dmaen \land (\dmaaddr \in KR) \rightarrow reset \ \}
\end{split}
\end{equation}
Similarly, LTL Specification for preventing DMA access to the attestation stack is defined as:
\begin{equation}\label{eq:dma-stackrule3}
\eqsize
\begin{split}
& \text{\bf G}: \ \{ \dmaen \land (\dmaaddr \in XS) \rightarrow reset \ \}
\end{split}
\end{equation}
Finally, invariant~(\ref{eq:dma-at}) specifies that DMA must be always disabled while $PC$ is in \sw region.
This prevents DMA controller from helping malware escape during attestation.
\begin{equation}\label{eq:dma-at}
\eqsize
\begin{split}
 & \text{\bf G}: \ \{ (PC \in CR) \land \dmaen \rightarrow reset \ \}
\end{split}
\end{equation}
\subsubsection{\textbf{Verified Model}}
Figure~\ref{fig:DMA_FSM} shows the FSM verified to comply with invariants~(\ref{eq:DMA_AC_rule}) to (\ref{eq:dma-at}).

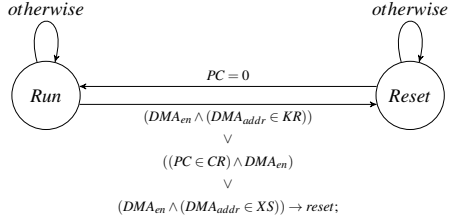
\begin{figure}[!hbtp]
\begin{center}
\noindent\resizebox{0.7\columnwidth}{!}{%
	\begin{tikzpicture}[->,>=stealth',auto,node distance=8.0cm,semithick]
		\tikzstyle{every state}=[minimum size=1.5cm]
		\tikzstyle{every node}=[font=\large]

		\node[state] 		(A)					{$Run$};
		\node[state]         (B) [right of=A,align=center]	{$Reset$};

		\path[->,every loop/.style={looseness=8}] 
			(A)
  				edge [loop above] node {$otherwise$} (A)
			(B)  
  				edge [loop above] node {$otherwise$} (B);
  		
\draw[->] (A.345) -- node[rotate=0,below, align=center,auto=right] {\small$(\dmaen \land (\dmaaddr \in KR))$ \\ \small$\lor$ \\ \small$((PC \in CR) \land \dmaen) $ \\ \small$\lor$ \\ \small$ (\dmaen \land (\dmaaddr \in XS)) \rightarrow reset;$} (B.195);
\draw[<-] (A.15) -- node[rotate=0,above] {\small$PC=0$} (B.165);
	\end{tikzpicture}
}
\caption{Verified FSM for DMA protection}
\label{fig:DMA_FSM}
\end{center}
\end{figure}

\subsection{\hw Composition}
Thus far, we designed and verified individual \hw sub-modules according to the methodology 
in Section~\ref{verif_metho} and illustrated in Figure~\ref{fig:module_verif}. We now follow the workflow of Figure~\ref{fig:composition_verif} 
to combine the sub-modules into a single Verilog module. Since each sub-module individually guarantees a 
subset of properties \textbf{P1--P7}, the composition is simple:  the system must reset whenever any sub-module 
reset is triggered. This is implemented by a logical OR of sub-modules reset signals. The composition is shown 
in Figure~\ref{fig:hwmod_fig}.

To verify that all LTL specifications still hold for the composition, we use 
Verilog2SMV~\cite{irfan2016verilog2smv} to translate \hw to SMV and verify SMV 
for all of these specifications simultaneously.

\begin{figure}
\centering
\includegraphics[width=0.5\columnwidth]{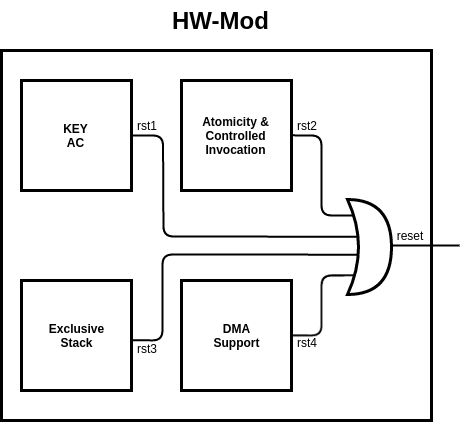}
\caption{\hw composition from sub-modules}\label{fig:hwmod_fig}
\end{figure}

\subsection{Secure Reset (\hw)}\label{sec_reset}
Finally, we define an LTL Specification for secure reset (\textbf{P3}).
According to Axiom \textbf{A4}, all registers (including $PC$) are set to $0$ on reset.
However, the reset routine implemented by the MCU might take several clock cycles.
Ensuring that $reset=1$ until the point when registers are wiped is important in order to guarantee that \attkey is not leaked through registers after a reset.
That is because some part of \attkey might remain in some of the registers if a reset happens during \sw execution.

\subsubsection{\textbf{LTL Specification}}
To guarantee that the reset signal is active for long enough so that the MCU reset finishes and all registers 
are cleaned-up, it must hold that:
\begin{equation}\label{eq:reset_LTL}
\eqsize
\begin{split}
& \text{\bf G}: \ \{ reset \rightarrow [(reset \quad \text{\bf U} \quad PC = 0) \quad \lor \quad \text{\bf G}(reset)] \ \}
\end{split}
\end{equation}
Invariant~(\ref{eq:reset_LTL}) states: when reset signal is triggered, it can only be released after $PC=0$.
Transition from $Reset$ state in all sub-modules presented in this section already takes 
this invariant into account. Thus, \hw composition also verifies LTL Specification~(\ref{eq:reset_LTL}).

\section{Alternative Designs}\label{alternatives}
%
%
%
We now discuss alternative designs for \acron that guarantee
verified properties without requiring a separate secure stack region for \sw operations.
Recall that \hw enforces that only \sw can access this stack. Since memory usage in HACL* \hmac is 
deterministic, the size of the separate stack can be pre-determined  -- $2,332$bytes. Even though resulting
in overall (HW and SW) design simplicity, dedicating $3\%$ of addressable memory to secure \RA might not be 
desirable. Therefore, we consider several alternatives.
In Section~\ref{eval} the costs involved with these alternatives are quantified and compared to the standard design of \acron.

\subsection{Erasure on \sw}%
\label{erasure_sw}
The most intuitive alternative to a reserved secure stack (which prevents accidental key leakage by \sw)
is to encode corresponding properties into the HACL* implementation and proof.
Specifically, it would require extending the HACL* implementation to zero out all allocated memory 
before every function return. In addition, to retain verification of \textbf{P2}  (in Section \ref{high_prop}) and 
ensure no leakage, HACL*-verified properties must be extended to incorporate memory erasure. 
This is not yet supported in HACL* and doing so would incur a slight 
performance overhead. However, the trade-off between performance and \ram savings might be worthwhile.

At the same time, we note that, even with verified erasure as a part of \sw, \textbf{P2} is still not guaranteed if 
the MCU does not guarantee erasure of the entire \ram upon boot. This is necessary in order to consider 
the case when \dev re-boots in the middle of \sw execution. Without a reserved stack, \attkey might persist in \ram.  Since the memory range for \sw execution is not fixed, hardware support is
required to bootstrap secure \ram erasure before starting any software execution.  In fact, such support is
necessary for all approaches without a separate secure stack.

\subsection{Compiler-Based Clean-Up}
\label{comp-clean}
While stack erasure in HACL* would integrate nicely with the overall proof of \sw, the assurance would be at 
the language abstraction level, and not necessarily at the machine level. The latter would require additional
assumptions about the compilation tool chain. We could also consider performing stack erasure 
directly in the compiler.  In fact, a recent proposal to do exactly that was made in {\sf zerostack}~\cite{SiChAn2018}, 
an extension to Clang/LLVM.   In case of \acron, this feature could be used on unmodified HACL* (at compilation time), 
to add instructions to erase the stack before the return of each function enabling~\textbf{P2}, assuming the existence 
of a verified \ram erasure routine upon boot.  We emphasize that this approach may increase the compiler's trusted code base. Ideally,
it should be implemented and formally verified as part of a verified compiler suite, such as CompCert~\cite{compcert}.

\begin{savenotes}
\begin{table*}[!hbtp]
\centering
\caption{Evaluation of cost, overhead, and performance of \RA}
\label{tab:hw-cost}

\resizebox{0.8\linewidth}{!}{%
\begin{tabular}{|l|c|ccc|c|c|cc|cc|}
\hline
\multirow{2}{*}{Method}                                           & \multirow{2}{*}{\begin{tabular}[c]{@{}c@{}}RAM Erasure\\ Required Upon Boot?\end{tabular}} & \multicolumn{4}{c|}{FPGA Hardware}               & \multirow{2}{*}{\begin{tabular}[c]{@{}c@{}}Verilog \\ LoC\end{tabular}} & \multicolumn{2}{c|}{Memory (byte)} & \multicolumn{2}{c|}{Time to attest 4KB} \\
                                                                  &                                                                                   & LUT       & Reg      & \multicolumn{2}{c|}{Cell} &                    & ROM           & Sec. RAM         & CPU cycles                & ms (at 8MHz)                   \\ \hline\hline
Core (Baseline)            						& N/A     & 1842      & 684      & \multicolumn{2}{c|}{3044} & 4034       & 0             & 0         & N/A        & N/A      \\
Secure Stack (Section \ref{verif})            	& No      & 1964      & 721      & \multicolumn{2}{c|}{3237} & 4621       & 4500          & 2332      & 3601216    & 450.15   \\ 
Erasure on \sw (Section \ref{erasure_sw})     	& Yes     & 1954      & 717      & \multicolumn{2}{c|}{3220} & 4516       & 4522          & 0         & 3613283    & 451.66   \\
Compiler-based Clean-up (Section \ref{comp-clean})~\footnote{As mentioned in Section~\ref{comp-clean}, 
there is no formally verified msp430 compiler capable of performing stack erasure. 
Thus, we estimate overhead of this approach by manually inserting code required for erasing the stack in \sw.} 
												& Yes     & 1954      & 717      & \multicolumn{2}{c|}{3220} & 4516       & 4522          & 0         & 3613283    & 451.66   \\
Double-\hmac Call (Section \ref{ssec:dmac})    & Yes     & 1954      & 717      & \multicolumn{2}{c|}{3220} & 4516       & 4570          & 0         & 7201605    & 900.20               \\
\hline
\end{tabular}
}
\end{table*}
\end{savenotes}

\subsection{Double-\hmac Call}
\label{ssec:dmac}
Finally, complete stack erasure could also be achieved directly using currently verified HACL* properties,
without any further modifications.  This approach involves invoking HACL* \hmac function a second time, 
after the computation of the actual \hmac. The second "dummy" call would use the same input data, however, 
instead of using \attkey, an independent constant, such as $\{0\}^{512}$, would be used as the \hmac key.  

Recall that HACL* is verified to only allocate memory on the stack in a deterministic manner. Also,
due to HACL*'s verified properties that mitigate side-channels, software flow does not change based on
the secret key. Therefore, this deterministic allocation implies that, for inputs of the 
same size, any variable allocated by the first "real"  \hmac call  (tainted by \attkey), would be overwritten 
by the corresponding variable in the second "dummy" call. Note that the same guarantee discussed in 
Section~\ref{erasure_sw} is provided here and secure \ram erasure at boot would still be needed for 
the same reasons. Admittedly, this double-HMAC approach would consume twice as many CPU cycles. 
Still, it might be a worthwhile trade-off, especially, if there is memory shortage and lack of previously 
discussed HACL* or compiler extension.

\section{Evaluation}\label{eval}
We now discuss implementation details and evaluate \acron's overhead and performance.
Section~\ref{vperf} reports on verification complexity.
Section~\ref{overhead} discusses performance in terms of time and space complexity as well as its hardware overhead.
We also provide a comparison between \acron and other \RA architectures targeting low-end devices, namely SANCUS~\cite{Sancus17} and SMART~\cite{smart}, in Section~\ref{sec:comparison}.

\subsection{Implementation}\label{impl-details}
As mentioned earlier, we use OpenMSP430~\cite{openmsp430} as an open core implementation of the MSP430 architecture. 
OpenMSP430 is written in the Verilog hardware description language (HDL) and can execute software generated by any 
MSP430 toolchain with near cycle accuracy.
We modified the standard OpenMSP430 to implement the hardware architecture presented in Section~\ref{sys-arch}, as shown 
in Figure~\ref{fig:arch}. This includes adding \rom to store \attkey and \sw, adding \hw, and adapting the memory backbone accordingly.
We use Xilinx Vivado~\cite{vivado} -- a popular logic synthesis tool -- to synthesize an RTL description of \hw into 
hardware in FPGA.
FPGA synthesized hardware consists of a number of logic cells. Each consists of Look-Up Tables (LUTs) and registers; 
LUTs are used to implement combinatorial boolean logic while registers are used for sequential logic elements, i.e., 
FSM states and data storage. We compiled \sw using the native msp430-gcc~\cite{msp430-gcc} and used Linker scripts
to generate software images compatible with the memory layout of Figure~\ref{fig:arch}.
Finally, we evaluated \acron on the FPGA platform targeting Artix-7~\cite{artix7} class of devices.
\begin{table}[!hbtp]
\centering
\caption{Verification results running on a desktop @ 3.40 GHz.}
\label{tab:nusmv-usage}
\resizebox{0.8\linewidth}{!}{%
\begin{tabular}{l|c|c|c|c}
      HW Submod. & LTL Spec. & Mem. (MB) & Time (s) & Verified \\ \hhline{=|=|=|=|=}
      Key AC & \ref{eq:AC_rule},\ref{eq:reset_LTL} & 7.5 & .02 & \cmark \\
      Atomicity & \ref{eq:at1},\ref{eq:at2},\ref{eq:irq},\ref{eq:reset_LTL} & 8.5 & .05 & \cmark \\
      Exclusive Stack & \ref{eq:stackrule1},\ref{eq:stackrule2},\ref{eq:reset_LTL} & 8.1 & .03 & \cmark \\
      DMA Support & \ref{eq:DMA_AC_rule}-\ref{eq:reset_LTL} & 8.2 & .04 & \cmark \\ \hline
      \hw & \ref{eq:AC_rule}-\ref{eq:reset_LTL} & 13.6       & .28  & \cmark   \\
\end{tabular}
}
\end{table}
\subsection{Verification Results}\label{vperf}

As discussed in Section~\ref{high_prop}, \acron's verification consists of properties \textbf{P1}--\textbf{P7}. 
\textbf{P5} is achieved directly by executing \sw from \rom. Meanwhile, HACL* \hmac verification implies \textbf{P4}.
All other properties are automatically verified using NuSMV model checker. 
Table~\ref{tab:nusmv-usage} shows the verification results of \acron's \hw composition as well as results for individual sub-modules.
It shows that \acron successfully achieves all the required security properties.
These results also demonstrate feasibility of our verification approach, since the verification process -- running on a commodity desktop 
computer -- consumes only small amount of memory and time:  $<14$MB and $0.3$sec, respectively, for all properties.

\subsection{Performance and Hardware Cost}\label{overhead}
We now report on \acron's performance considering the standard design (described in Section~\ref{verif}) and alternatives 
discussed in Section~\ref{alternatives}.
We evaluate the hardware footprint,  memory (\rom and secure \ram), and run-time. 
Table~\ref{tab:hw-cost} summarizes the results.\\
\noindent\textbf{Hardware Footprint.}
The secure stack approach adds around 434 lines of code in Verilog HDL. 
This corresponds to around 20\% of the code in the original OpenMSP430 core.
In terms of synthesized hardware, it requires 122 (6.6\%) and 19 (5.4\%) additional LUTs and registers respectively. 
Overall, \acron contains 193 logic cells more than the unmodified OpenMSP430 core, corresponding to a 6.3\% increase.\\
%
%
\noindent\textbf{Memory.}
\acron requires $\sim$4.5KB of \rom; 
most of which (96\%) is for storing HACL* \hmac -SHA256 code.
The secure stack approach has the smallest \rom size, as it does not need to perform a memory clean-up in software.
However, this advantage is attained at the price of requiring 2.3KBytes of reserved \ram.
This overhead corresponds to 3.5\% of MSP430 16-bit address space.\\
\noindent\textbf{Attestation Run-time.}
Attestation run-time is dominated by the time it takes to compute the \hmac of \dev's memory.
The secure stack, erasure on \sw and compiler-based clean-up approaches take roughly .45$s$ to attest 
4$KB$ of \ram on an MSP430 device with a clock frequency at 8MHz.
Whereas, the double MAC approach requires invoking the \hmac function twice, leading its run-time to be roughly two times slower.\\
\noindent\textbf{Discussion.}
We consider \acron's overhead to be affordable.
The additional hardware, including registers, logic gates and exclusive memory, resulted in only a 3-6\% increase.
The number of cycles required by \sw exhibits a linear increase with the size of attested memory.
As MSP430 typically runs at 8-25MHz, attestation of the entire \ram on a typical MSP430 can be computed in less than a second.
\acron's \RA is relatively cheap to the \dev.
As a point of comparison we can consider a common cryptographic primitive such as the Curve25519 Elliptic-Curve Diffie-Hellman (ECDH) key exchange.
A single execution of an optimized version of such protocol on MSP430 has been reported to take $\approx 9$ million cycles~\cite{hinterwalder2014full}.
As  Table~\ref{tab:hw-cost} shows, attestation of $4$KBytes (typical size of \ram in some MSP430 models) can be computed three times faster.

\subsection{Comparison with Other Low-End \RA Architectures}\label{sec:comparison}

We here compare \acron's overhead with two widely known \RA architectures targeting low-end embedded systems: SMART~\cite{smart} and SANCUS~\cite{Sancus17}.
We emphasize, however, that both SMART and SANCUS were designed in an ad hoc manner.
Thus, they can not be formally verified and do not provide any guarantees offered by \acron's verified architecture.
Nevertheless, it is considered important to contrast \acron's cost with such architectures to demonstrate its affordability.

Table~\ref{tab:comparison} presents a comparison between features offered and required by aforementioned architectures.
SANCUS is, to the best of our knowledge, the cheapest pure HW-based architecture, while SMART is a minimal HW/SW \RA co-design.
Since SANCUS's \RA routine is implemented entirely in HW, it does not require \rom to store the SW implementation of the integrity ensuring function.
\acron implements a MAC with digest sizes of 256-bits.
\smart and SANCUS, on the other hand, use SHA1-based MAC and SPONGNET-128/128/8~\cite{bogdanov2013spongent}, respectively.
Such MACs do not offer strong collision resistance due to the small digest sizes (and known collisions).
Of the three architectures, \acron is the only one secure in the presence of DMA
and the only one to be rigorously specified and formally verified.

\begin{table}[!t]
\centering
\caption{Qualitative comparison between \RA architectures targeting low-end devices}
\label{tab:comparison}
\resizebox{\linewidth}{!}{%
\begin{tabular}{|l|c|c|c|}
                   & \acron & SMART & SANCUS \\ \hhline{=|=|=|=|} 
      Design Type & Hybrid (HW/SW) & Hybrid (HW/SW) & Pure HW \\
      \RA function & HMAC-SHA256 & HMAC-SHA1 & SPONGNET-128/128/8 \\
      \rom for \RA code & Yes & Yes & No  \\
      DMA Support & Yes & No & No \\
      Formally Verified & Yes & No & No  \\
\end{tabular}
}
\end{table}

Figure~\ref{fig:comparison} presents a quantitative comparison between the \RA architectures.
It considers additional overhead in relation to the latest version of the unmodified OpenMSP430 (Available at~\cite{openmsp430}).
Compared to \acron, SANCUS requires $12\times$ more Look-Up Tables, $22\times$ more registers, and its (unverified) TCB
is 2.5 times larger in lines of Verilog code. This comparison demonstrates the cost of relying on a HW-only approach even when designed for minimality.
\smart's overhead is slightly smaller than that of \vrased due to lack of DMA support.
In terms of attestation execution time, \smart is the slowest, requiring 9.2M clock cycles to attest 4KB of memory.
SANCUS achieves the fastest attestation time (1.3M cycles) due to the HW implementation of SPONGNET-128/128/8.
\acron sits in between the two with a total attestation time of 3.6M cycles.

\begin{figure}[t]
	\centering
	\subfigure[Additional HW overhead (\%) in Number of Look-Up Tables]
	{\includegraphics[width=0.45\columnwidth]{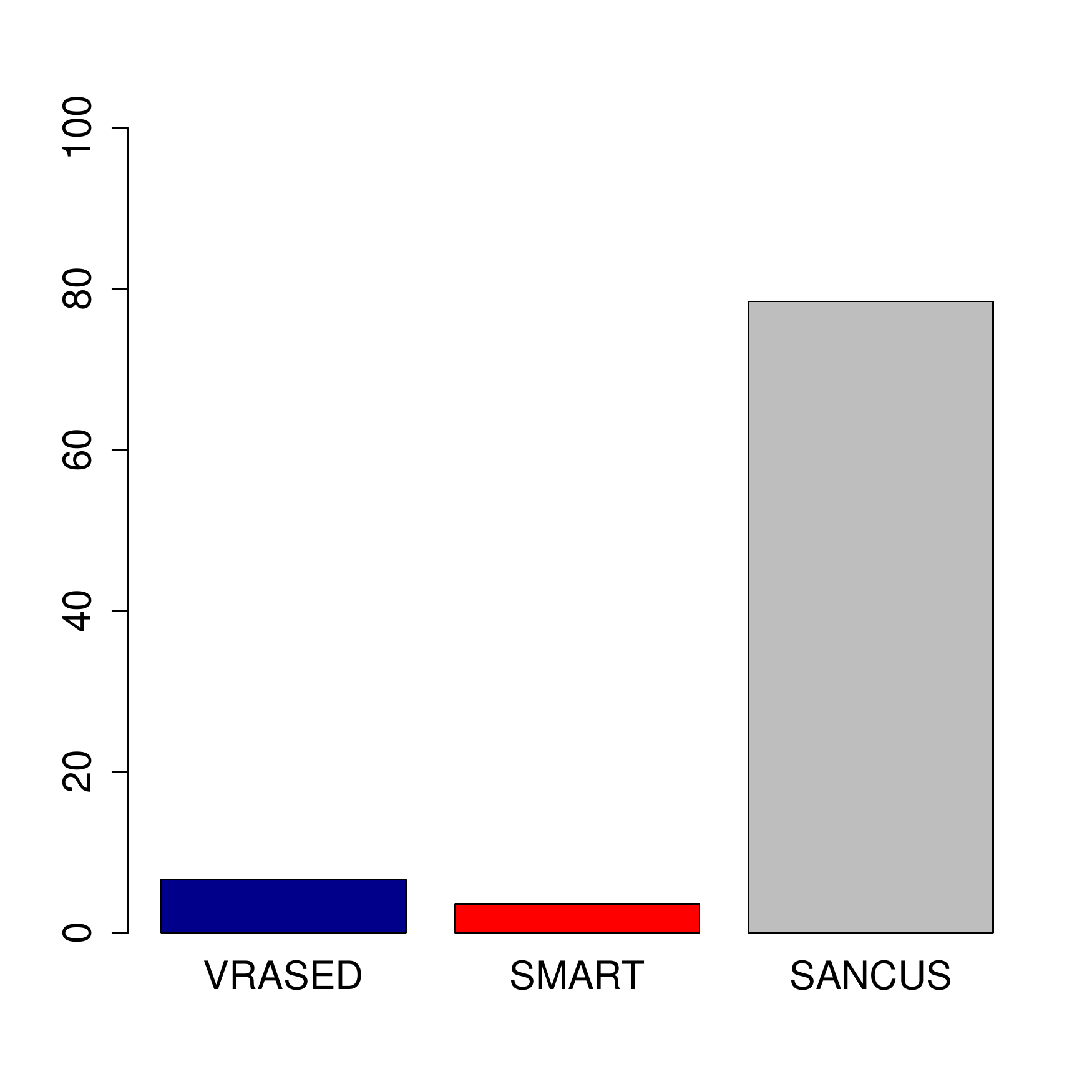}}
	\subfigure[Additional HW overhead (\%) in Number of Registers]
	{\includegraphics[width=0.45\columnwidth]{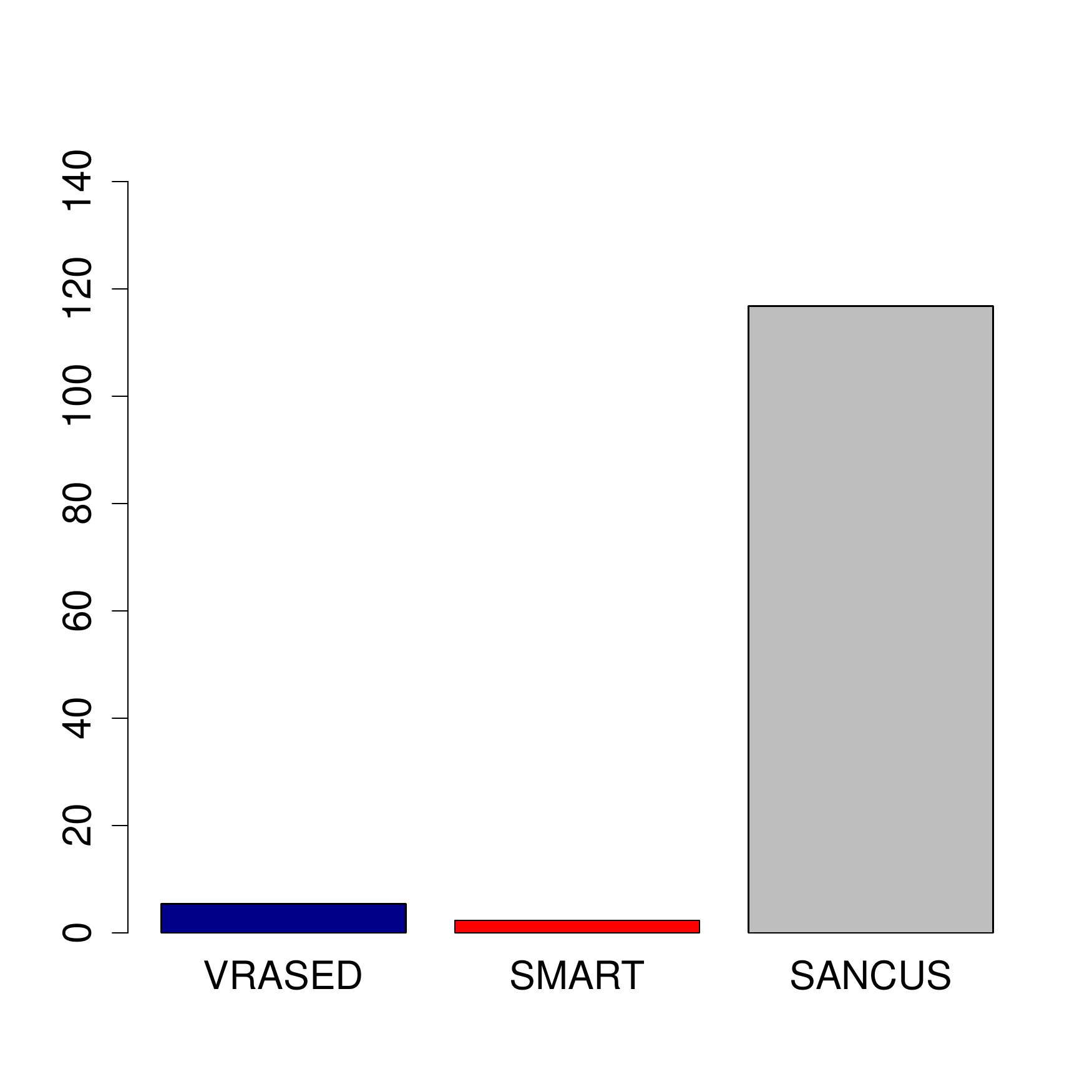}}
	\subfigure[Additional Verilog Lines of Code]
	{\includegraphics[width=0.45\columnwidth]{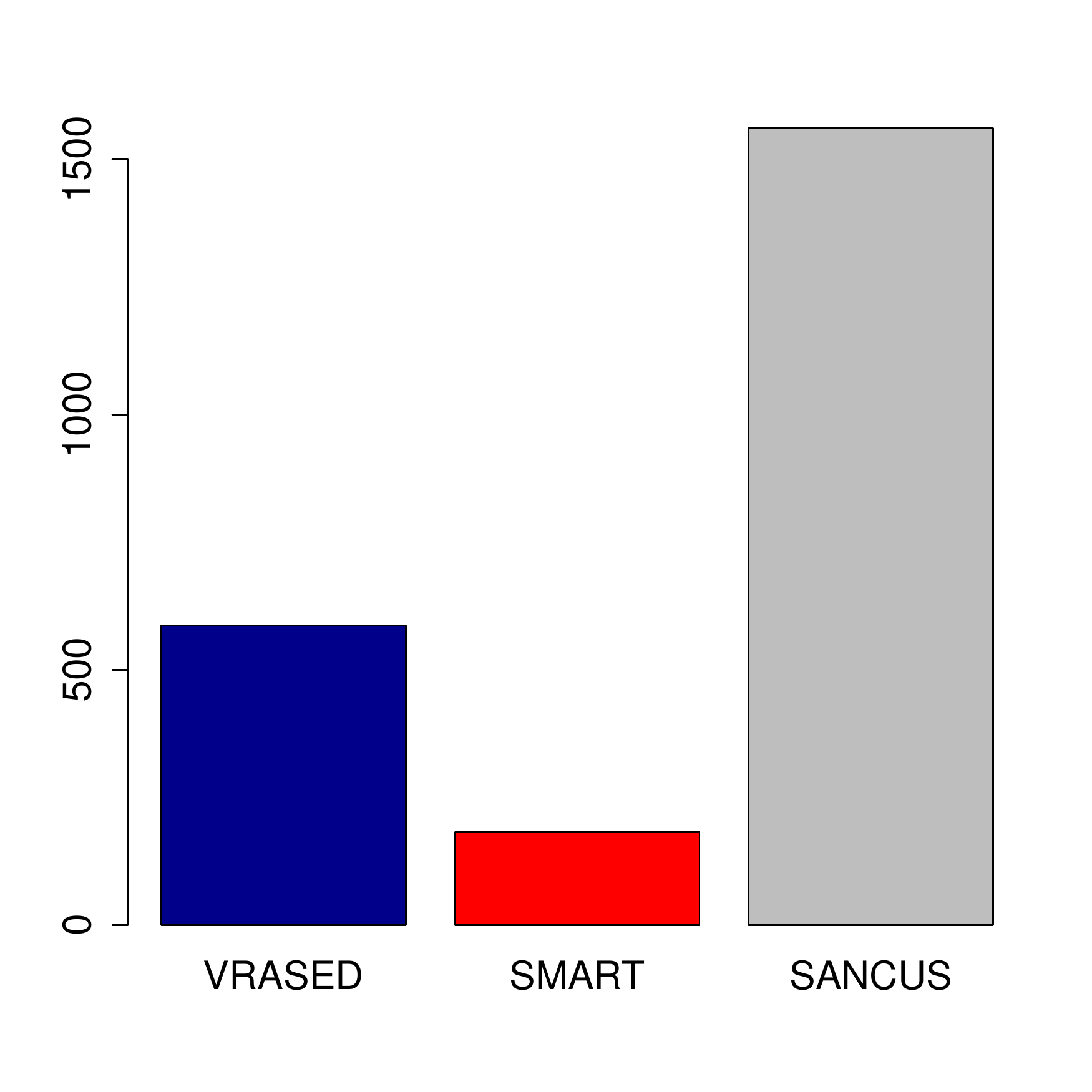}}
	\subfigure[Time to attest 4KB (in millions of CPU cycles)]
	{\includegraphics[width=0.45\columnwidth]{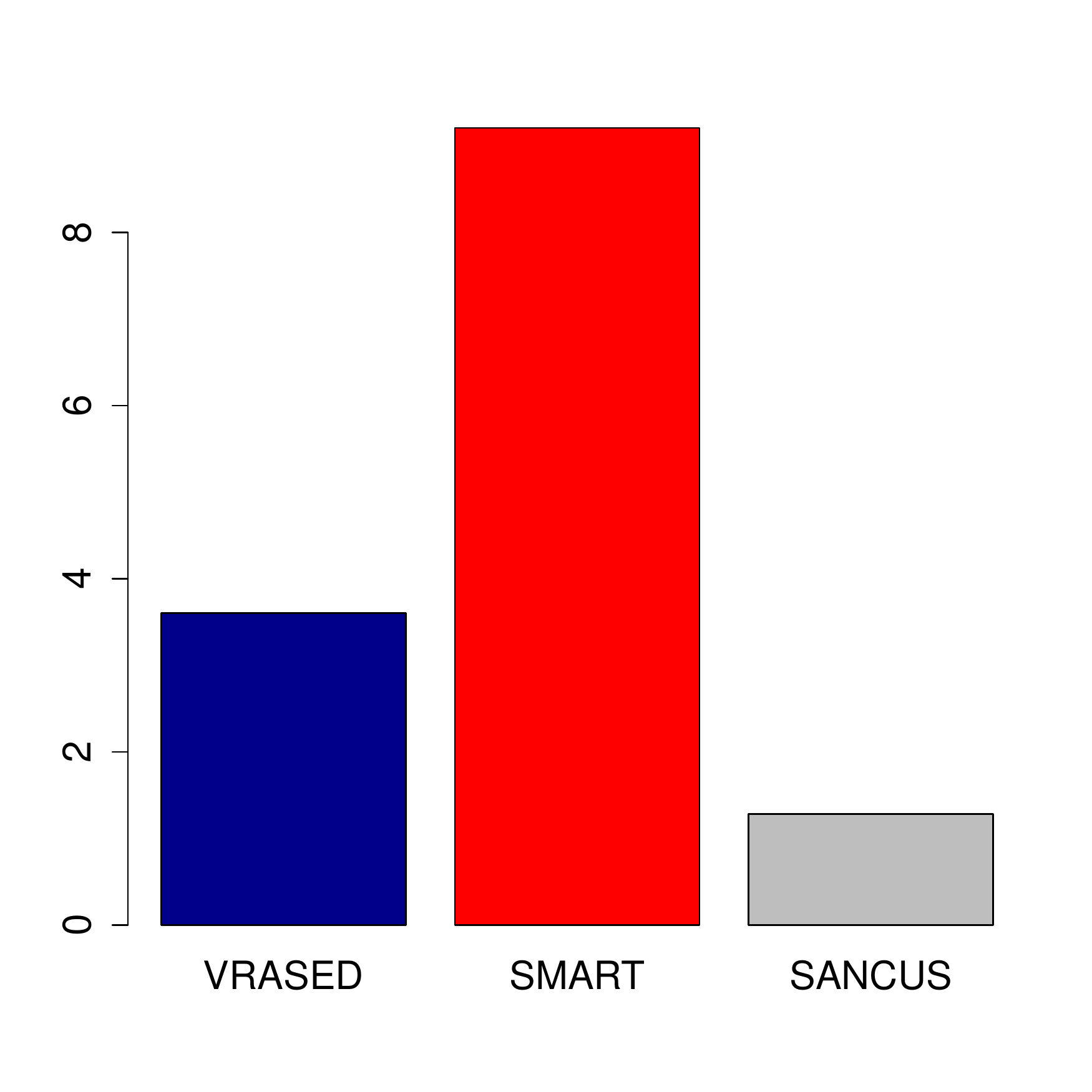}}
	\caption{Comparison between \RA architectures targeting low-end devices}\label{fig:comparison}
\end{figure}
\section{Related Work}\label{RW}
We are unaware of any previous work that yielded a formally verified \RA design (\RA architectures are overviewed in Section~\ref{sec:ra_bg}).
To the best of our knowledge, \acron is the first verification of a security service implemented as HW/SW co-design.
Nevertheless, formal verification has been widely used as the \textit{de facto} means to guarantee that a system is 
free of implementation errors and bugs. In recent years, several efforts focused on verifying security-critical systems.

In terms of cryptographic primitives, Hawblitzel et al. \cite{hawblitzel2014ironclad} verified new implementations of 
SHA, HMAC, and RSA. Beringer et al.\cite{beringer2015verified} verified the Open-SSL 
SHA-256 implementation. Bond et al.~\cite{bond2017vale} verified an assembly implementation of SHA-256, Poly1305, 
AES and ECDSA. More recently, Zinzindohou{\'e}, et al.~\cite{hacl} developed HACL*, a verified cryptographic 
library containing the entire cryptographic API of NaCl~\cite{bernstein2012security}. As discussed earlier, HACL*'s verified \hmac
forms the core of \acron's software component.

Larger security-critical systems have also been successfully verified. For example, Bhargavan~\cite{bhargavan2013implementing} 
implemented the TLS protocol with verified cryptographic security. CompCert\cite{compcert} is a C compiler that is formally verified to 
preserve C code semantics in generated assembly code. Klein et al.~\cite{sel4} designed and proved functional correctness of \sel\  -- 
the first verified general-purpose microkernel. More recently, Tuncay et al. verified a design for Android OS App permissions 
model~\cite{tuncay2018resolving}.

The importance of verifying \RA has been recently acknowledged by Lugou et al.~\cite{lugoutoward}, which discussed methodologies 
for specifically verifying HW/SW \RA co-designs. A follow-on result proposed the SMASH-UP tool~\cite{lugou2017smashup}. 
By modeling a hardware abstraction, SMASH-UP allows automatic conversion of assembly instructions to the effects on hardware representation.
Similarly, Cabodi et al.~\cite{cabodi2015formal,cabodi2016secure} discussed the first steps towards formalizing hybrid \RA properties.
However, none of these results yielded a fully verified (and publicly available) \RA architecture, such as \acron.

\section{Conclusion}\label{conclusion}
This paper presents \acron ~-- the first formally verified \RA method that uses a verified cryptographic software 
implementation and combines it with a verified hardware design to guarantee correct implementation of \RA security properties.
\acron is also the first verified security service implemented as a HW/SW co-design.
\acron was designed with simplicity and minimality in mind. It results in efficient computation and low
hardware cost, realistic even for low-end embedded systems.
\acron's practicality is demonstrated via publicly available implementation using the low-end MSP430 platform.
The design and verification methodology presented in this paper can be extended to other MCU architectures.
We believe that this work represents an important and timely advance in embedded systems security, especially, 
with the rise of heterogeneous ecosystems of (inter-)connected IoT devices. 

The most natural direction for future work is to adapt \acron to other MCU architectures.
Such an effort could follow the same verification methodology presented in this paper.
It would involve: (1) mapping MCUs specifications to a set of axioms (as we did for MSP430 in Section~\ref{approach}), and 
(2) adapting the proofs by modifying the LTL Specifications and hardware design (as in in Section~\ref{verif}) accordingly.
A second direction is to extend \acron's capabilities to include and verify other
trusted computing services such as secure updates, secure deletion, and remote code execution.
It would also be interesting to verify and implement other \RA
designs with different requirements and trade-offs, such as software-
and hardware-based techniques.
In the same vein, one promising
direction would be to verify \hydra~\RA architecture~\cite{hydra},
which builds on top of the formally verified \sel~\cite{sel4}
microkernel.
Finally, the optimization of \vrased's \hmac, with respect to computation and memory allocation, while retaining its verified properties, is an interesting open problem.

\bibliographystyle{abbrv}

{\footnotesize
\linespread{0}
\bibliography{references}
}

\vspace{1em}
\begin{large}
 \noindent\textbf{APPENDIX}
\end{large}

\appendix

\section{\RA Soundness and Security Proofs}\label{proofs}

\subsection{Proof Strategy}
We present the proofs for \RA soundness (Definition~\ref{ra_soundness}) and \RA security (Definition \ref{sec_def}).
Soundness is proved entirely via LTL equivalences.
In the proof of security we first show, via LTL equivalences, that \acron guarantees that adversary \adv\ can never learn \attkey with more than negligible probability.
We then prove security by showing a reduction from \hmac existential unforgeability to security.
In other words, we show that existence of \adv\ that breaks \vrased implies existence of \hmac-\adv\ able to break conjectured existential unforgeability of \hmac.
The full machine-checked proofs for the LTL equivalences (using Spot 2.0~\cite{spot} proof assistant) discussed in the remainder of this section are available in~\cite{public-code}.

\subsection{Machine Model}
To prove that \acron's design satisfies end-to-end definitions of soundness and security for \RA, we start by formally defining (in LTL) memory and execution models corresponding to the architecture introduced in Section~\ref{approach}.
\begin{figure}[!ht]
\begin{mdframed}
\scriptsize
\begin{definition}[Memory model]\label{mem_model}~\\ 
\begin{enumerate}
 \item \attkey is stored in \rom $\leftrightarrow$ $G:\{KR = \attkey\}$ 
 \item \swtiny is stored in \rom $\leftrightarrow$ $G:\{CR = \swtiny\}$ 
 \item $MR$, $CR$, $AR$, $KR$, and $XS$ are non-overlapping memory regions
 \end{enumerate}
\end{definition}
\end{mdframed}
\end{figure}

The memory model in Definition~\ref{mem_model} captures that $KR$ and $CR$ are \rom regions, and are thus immutable. Hence, the values stored in those regions always correspond to \attkey and \sw code, respectively.
Finally, the memory model states that $MR$, $CR$, $AR$, $KR$, and $XS$ are disjoint regions in the memory layout, corresponding to the architecture in Figure~\ref{fig:arch}.
\begin{figure}[!ht]
\begin{mdframed}
\scriptsize
\begin{definition}[Execution model]\label{exec_model}~\\ 
\begin{enumerate}
 \item Modify\_Mem($i$) $\rightarrow$ $(W_{en} \land D_{addr} = i) \lor (DMA_{en} \land DMA_{addr} = i)$
 \item Read\_Mem($i$) $\rightarrow~ (R_{en} \land D_{addr} = i) \lor (DMA_{en} \land DMA_{addr} = i)$
 \item Interrupt $\rightarrow$ $irq$
\end{enumerate}
\end{definition}
\end{mdframed}
\end{figure}

Our execution model, in Definition~\ref{exec_model}, translates MSP430 behavior by capturing the effects on the processor signals when reading and writing from/to memory.
We do not model the effects of instructions that only modify register values (e.g., ALU operations, such as $add$ and $mul$) because they are not necessary in our proofs.

The execution model defines that a given memory address can be modified in two cases: by a CPU instruction or by DMA.
In the first case, the $W_{en}$ signal must be on and $D_{addr}$ must contain the memory address being accessed.
In the second case, $DMA_{en}$ signal must be on and $DMA_{addr}$ must contain the address being modified by DMA.
The requirements for reading from a given address are similar, except that instead of $W_{en}$, $R_{en}$ must be on.
Finally, the execution model also captures the fact that an interrupt implies setting the $irq$ signal to 1. 

\subsection{\RA Soundness Proof}

The proof follows from \sw functional correctness (expressed by Definition~\ref{sw_fc}) and LTL specifications~\ref{eq:at1}, \ref{eq:irq}, \ref{eq:stackrule2}, and \ref{eq:dma-at}

\begin{theorem}\label{soundeness_th}
VRASED is sound according to Definition~\ref{ra_soundness}.
\end{theorem}
\begin{proof}
\scriptsize
\begin{align*}
	\scriptsize
	\label{eq:lemma1}
	Definition \ \ref{sw_fc} \ \land \ LTL_{\ref{eq:at1}} \land LTL_{\ref{eq:irq}} \land LTL_{\ref{eq:stackrule2}} \land LTL_{\ref{eq:dma-at}} \rightarrow Theorem~\ref{ra_soundness}
\end{align*}
\end{proof}
The formal computer proof for Theorem~\ref{ra_soundness} can be found in~\cite{public-code}.
Due to space limitations, we only provide some intuition, by splitting the proof into two parts.
First, \sw functional correctness (Definition~\ref{sw_fc}) would imply Theorem~\ref{ra_soundness} if
$AR$, $CR$, $KR$ never change and an interrupt does not happen during \sw computation. However, memory model Definitions~\ref{mem_model}.1 and~\ref{mem_model}.2 already
guarantee that $CR$ and $KR$ never change. Also, LTL~\ref{eq:irq} states that an interrupt cannot happen during \sw computation, otherwise the device resets.
Therefore, it remains for us to show that $AR$ does not change during \sw computation.
This is stated in Lemma~\ref{lemma1l}.
\begin{figure}[!ht]
\begin{mdframed}
\scriptsize
\begin{lemma}{Temporal Consistency -- Attested memory does not change during \swtiny computation}\label{lemma1l}
\begin{align*}
\begin{split}
	& \textbf{G}: \ \{ \\
    	& PC = CR_{min} \land AR = M \land \neg reset \ \textbf{U} \ (PC=CR_{max}) \rightarrow \\
	& (AR = M) \ \textbf{U} \ (PC=CR_{max}) \ \}
\end{split}
\end{align*}
\end{lemma}
\end{mdframed}
\end{figure}

In turn, Lemma~\ref{lemma1l} can be proved by:
\begin{equation}\label{eq:lemma1}
	\scriptsize
	 LTL_{\ref{eq:at1}} \land LTL_{\ref{eq:stackrule2}} \land LTL_{\ref{eq:dma-at}} \rightarrow Lemma~\ref{lemma1l}
\end{equation}

The reasoning for Equation~\ref{eq:lemma1} is as follows: 
\begin{compactitem}
\item $LTL_{\ref{eq:at1}}$ prevents the CPU from stopping execution of \sw before its last instruction.
\item $LTL_{\ref{eq:stackrule2}}$ guarantees that the only memory regions written by the CPU during \sw execution are $XS$ and $MR$, which do not overlap with $AR$.
\item $LTL_{\ref{eq:dma-at}}$ prevents DMA from writing to memory during \sw execution.
\end{compactitem}
Therefore, there are no means for modifying $AR$ during \sw execution, implying Lemma~\ref{lemma1l}.
As discussed above, it is easy to see that:
\begin{equation}
	\scriptsize
	Lemma~\ref{lemma1l} \land LTL_{\ref{eq:irq}} \land Definition~\ref{sw_fc} \rightarrow Theorem~\ref{ra_soundness}
\end{equation}
\subsection{\RA Security Proof}
Recall the definition of \RA security in the game in Figure~\ref{fig:security_game}.
The game makes two key assumptions:
\begin{enumerate}
\item \sw call results in a temporally consistent \hmac of $AR$ using a key derived from \attkey and $\chal$. This is already proved by \acron's soundness.
\item \adv\ never learns \attkey with more than negligible probability.
\end{enumerate}
\begin{figure}[!ht]
\begin{mdframed}
\scriptsize
\begin{lemma}{Key confidentiality -- \attkey can not be accessed directly by untrusted software ($\neg(PC \in CR)$) and any memory written to by \swtiny can never be read by untrusted software.}\label{lemma:key-leak}
\begin{align*}
\begin{split}
& \textbf{G}:\{ \\
& (\neg(PC \in CR) \land Read\_Mem(i) \land i \in KR \rightarrow reset) \land \\
& (\dmaen \land DMA_{addr} = i \land i \in KR \rightarrow reset) \land \\
& [\neg reset \land PC \in CR \land Modify\_Mem(i) \land \neg(i \in MR) \rightarrow \\
& \textbf{G}:\{(\neg(PC \in CR) \land Read\_Mem(i) \lor \dmaen \land DMA_{addr} = i)\\
& \rightarrow reset\}] \\
& \}
\end{split}
\end{align*}
\end{lemma}
\end{mdframed}
\end{figure}
By proving that \acron's design satisfies assumptions 1 and 2, we show that the capabilities of untrusted software (any DMA or CPU software other than \sw) on \dev are equivalent to the capabilities of \adv\ in \RA-game.
Therefore, we still need to prove item 2 before we can use such game to prove \acron's security.
The proof of \adv's inability to learn \attkey with more than negligible probability is facilitated by \textit{A6 - Callee-Saves-Register} convention stated in Section~\ref{approach}.
A6 directly implies no leakage of information through registers on the return of \sw. This is because, before the return of a function, registers must be restored to their state prior to the function call.
Thus, untrusted software can only learn \attkey (or any function of \attkey) through memory.
However, if untrusted software can never read memory written by \sw, it never learns anything about \attkey (not even through timing side channels since \sw is secret independent).
Now, it suffices to prove that untrusted software can not access \attkey directly and that it can never read memory written by \sw. These conditions are stated in LTL in Lemma~\ref{lemma:key-leak}.
We prove that \acron satisfies Lemma~\ref{lemma:key-leak} by writing a computer proof (available in~\cite{public-code}) for Equation~\ref{proof_lemma2}.
The reasoning for this proof is similar to that of \RA soundness and omitted due to space constraints.
\begin{equation}\label{proof_lemma2}
\scriptsize
LTL_{\ref{eq:AC_rule}} \land LTL_{\ref{eq:stackrule1}} \land LTL_{\ref{eq:stackrule2}} \land LTL_{\ref{eq:DMA_AC_rule}} \land LTL_{\ref{eq:dma-stackrule3}} \land LTL_{\ref{eq:dma-at}} \rightarrow \text{Lemma 2}
\end{equation}
We emphasize that Lemma~\ref{lemma:key-leak} does not restrict reads and writes to $MR$, since this memory is used for inputting \chal and receiving \sw result.
Nonetheless, the already proved \RA soundness and LTL~\ref{eq:at2} (which makes it impossible to execute fractions of \sw) guarantee that $MR$ will not leak anything, because at the end of \sw computation it will always contain an \hmac result, which does not leak information about \attkey.
After proving Lemma~\ref{lemma:key-leak}, the capabilities of untrusted software on \dev are equivalent to those of adversary \adv\ in \RA-game of Definition~\ref{sec_def}.
Therefore, in order to prove \acron's security, it remains to show a reduction from \hmac security according to the game in Definition~\ref{sec_def}. \vrased's security is stated and proved in Theorem~\ref{reduction}. 
\begin{theorem}\label{reduction}
\acron is secure according to Definition~\ref{sec_def} as long as \hmac is a secure \mac.
\end{theorem}
\begin{proof}
\textit{A \mac\ is defined as tuple of algorithms $\{\texttt{Gen}, \texttt{Mac},\texttt{Vrf}\}$.
For the reduction we construct a slightly modified \hmac{$'$}, which has the same \texttt{Mac} and \texttt{Vrf} algorithms as standard \hmac\, but $\texttt{Gen} \leftarrow KDF(\attkey, \chal)$ where $\chal \leftarrow\$\{0,1\}^l$.
Since $KDF$ function itself is implemented as a \texttt{Mac} call, it is easy to see that the outputs of \texttt{Gen} are indistinguishable from random. In other words, the security of this slightly modified construction follows from the security of \hmac itself.
Assuming that there exists \adv\ such that $Pr[\adv,\RA_{game}] > negl(l)$, we show that such adversary can be used to construct \hmac-\adv\ that breaks existential unforgeability of \hmac' with probability $Pr[$\hmac-\adv$,$\mac-game$] > negl(l)$.  To that purpose \hmac-\adv\ behaves as follows:
\begin{enumerate}
\item \hmac-\adv\ selects $msg$ to be the same $M \neq AR$ as in \RA-game and asks \adv\ to produce the same output used to win \RA-game.
\item \hmac-\adv\ outputs the pair ($msg$,$\sigma$) as a response for the challenge in the standard existential unforgeability game, where $\sigma$ is the output produced by \adv\ in step 1.
\end{enumerate}
By construction, ($msg$,$\sigma$) is a valid response to a challenge in the existential unforgeability \mac\ game considering \hmac{$'$} as defined above. Therefore, \hmac-\adv\ is able to win the existential unforgeability game with the same $> negl(l)$ probability that \adv\ has of winning \RA-game in Definition~\ref{sec_def}.}
\end{proof}

\section{Optional Verifier Authentication\label{vrf_auth}}
\begin{figure}[!hbtp]
\begin{lstlisting}[basicstyle=\tiny, numberstyle=\tiny]
void Hacl_HMAC_SHA2_256_hmac_entry() {
    uint8_t key[64] = {0};
    uint8_t verification[32] = {0};
    if (memcmp(CHALL_ADDR, CTR_ADDR, 32) > 0)
    {
	memcpy(key, KEY_ADDR, 64);
	
	hacl_hmac((uint8_t*) verification, (uint8_t*) key,
		  (uint32_t) 64, *((uint8_t*)CHALL_ADDR) ,
		  (uint32_t) 32);
		  
	if (!memcmp(VRF_AUTH, verification, 32)
	{
	    hacl_hmac((uint8_t*) key, (uint8_t*) key, 
	    	(uint32_t) 64, (uint8_t*) verification, 
	    	(uint32_t) 32);
	    hacl_hmac((uint8_t*) MAC_ADDR, (uint8_t*) key, 
	    	(uint32_t) 32, (uint8_t*) ATTEST_DATA_ADDR, 
	    	(uint32_t) ATTEST_SIZE);
	    memcpy(CTR_ADDR, CHALL_ADDR, 32);
	}
    }

    return();
}
\end{lstlisting}
\caption{\sw Implementation with \vrf authentication}\label{fig:sw_att_code_auth}
\end{figure}
Depending on the setting where \dev is deployed, authenticating the attestation request before executing \sw may be required.
For example, if \dev is in a public network, the adversary may try to communicate with it.
In particular, the adversary can impersonate \vrf and send fake attestation requests to \dev, attempting to cause denial-of-service.
This is particularly relevant if \dev is a safety-critical device.
If \dev receives too many attestation requests, regular (and likely honest) software running on \dev would not execute 
because \sw would run all the time. Thus, we now discuss an optional part of \acron's design suitable for such settings.
It supports and formally verifies authentication of \vrf as part of \sw execution.
Our implementation is based on the protocol in~\cite{brasser2016remote}.

Figure~\ref{fig:sw_att_code_auth} presents an implementation of \sw that includes \vrf authentication.
It also builds upon HACL* verified HMAC to authenticate \vrf, in addition to computing the authenticated integrity check. 
In this case, \vrf's request additionally contains an HMAC of the challenge computed using \attkey.
Before calling \sw, software running on \dev is expected to store the received challenge on a fixed address $CHALL\_ADDR$ and the corresponding received HMAC on $VRF\_AUTH$.
\sw discards the attestation request if (1) the received challenge is less than or equal to the latest challenge, or (2) HMAC of the received challenge is mismatched.
After that, it derives a new unique key using HKDF~\cite{krawczyk2010hmac} from \attkey and the received \hmac
and uses it as the attestation key. 

\hw must also be slightly modified to ensure security of \vrf's authentication.
In particular, regular software must not be able modify the memory region that stores \dev's counter.
Notably, the counter requires persistent and writable storage, because \sw needs to modify it at the end of each attestation execution. Therefore, $CTR$ region resides on FLASH.
We denote this region as:
\begin{itemize}
\item $CTR = [CTR_{min}, CTR_{max}]$;
\end{itemize}

LTL Specifications~(\ref{eq:AC_rule_auth}) and~(\ref{eq:DMA_AC_rule_auth}) must hold (in addition to the ones discussed in Section~\ref{verif}).
\begin{equation}\label{eq:AC_rule_auth}
\eqsize
\begin{split}
\text{\bf G}: \ \{
\neg (PC \in CR) \land W_{en} \land (D_{addr} \in CTR) \rightarrow reset \ \}
\end{split}
\end{equation}
\begin{equation}\label{eq:DMA_AC_rule_auth}
\eqsize
\begin{split}
& \text{\bf G}: \ \{ \dmaen \land (\dmaaddr \in CTR) \rightarrow reset \}
\end{split}
\end{equation}
LTL Specification~(\ref{eq:AC_rule_auth}) ensures that regular software does not modify \dev's counter, while (\ref{eq:DMA_AC_rule_auth}) ensures that the same is not possible via the DMA controller.
FSMs in Figures~\ref{fig:AC_FSM} and~\ref{fig:DMA_FSM}, corresponding to \hw access control and DMA sub-modules, must be modified to transition into $Reset$ state according to these new conditions.
In addition, LTL Specification~(\ref{eq:stackrule2}) must be relaxed to allow \sw to write to $CTR$.
Implementation and verification of the modified version of these sub-modules are publicly available at \acron's repository~\cite{public-code} as an optional part of the design.

\section{VRASED API}\label{API}
\lstset{language=C,
	basicstyle={\scriptsize\ttfamily},
	showstringspaces=false,
	frame=single,
	xleftmargin=2em,
	framexleftmargin=3em,
	numbers=left, 
	numberstyle=\tiny,
	commentstyle={\tiny\itshape},
	keywordstyle={\tiny\bfseries},
	keywordstyle=\color{blue}\tiny\ttfamily,
	stringstyle=\color{red}\tiny\ttfamily,
        commentstyle=\color{black}\tiny\ttfamily,
        morecomment=[l][\color{magenta}]{\#},
        breaklines=true
}
\begin{figure} \scriptsize
\begin{lstlisting}[basicstyle=\tiny, numberstyle=\tiny, xleftmargin=.05\textwidth, xrightmargin=.05\textwidth]
void VRASED (uint8_t *challenge, uint8_t *response) {
    //Copy input challenge to MAC_ADDR:
    memcpy ( (uint8_t*)MAC_ADDR, challenge, 32);

    //Disable interrupts:
    __dint();

    //Save current value of r5 and r6:
    __asm__ volatile("push    r5" "\n\t");
    __asm__ volatile("push    r6" "\n\t");

    //Write return address of Hacl_HMAC_SHA2_256_hmac_entry
    //to RAM:
    __asm__ volatile("mov    #0x000e,   r6" "\n\t");
    __asm__ volatile("mov    #0x0300,   r5" "\n\t");
    __asm__ volatile("mov    r0,        @(r5)" "\n\t");
    __asm__ volatile("add    r6,        @(r5)" "\n\t");

    //Save the original value of the Stack Pointer (R1):
    __asm__ volatile("mov    r1,    r5" "\n\t");

    //Set the stack pointer to the base of the exclusive stack:
    __asm__ volatile("mov    #0x1000,     r1" "\n\t");

    //Call SW-Att:
    Hacl_HMAC_SHA2_256_hmac_entry();

    //Copy retrieve the original stack pointer value:
    __asm__ volatile("mov    r5,    r1" "\n\t");

    //Restore original r5,r6 values:
    __asm__ volatile("pop   r6" "\n\t");
    __asm__ volatile("pop   r5" "\n\t");

    //Enable interrupts:
    __eint();

    //Return the HMAC value to the application:
    memcpy(response, (uint8_t*)MAC_ADDR, 32);
}
\end{lstlisting}
\caption{\acron's wrapper function.}\label{fig:API}
\end{figure}

\acron ensures that any violation of secure \RA properties is detected and causes the system to reset.
However, benign applications running on the MCU must also comply with \acron rules to execute successfully.

For example, suppose that a benign application receives an attestation request from \vrf.
It then needs to set up MCU software state before \sw can execute.
This includes: disabling interrupts, setting the stack pointer to the reserved secure stack, and storing the previous stack pointer,
in order to restore software state after \sw execution completes. If the application fails to do this, even though \RA security still holds, 
the system might reach an unexpected state due to incorrect set-up before or after \sw execution. For example, suppose that
interrupts were (erroneously) not disabled before calling \sw, and an interrupt occurs during \sw execution, thus aborting the application. 
Though not a violation of secure \RA properties, this can clearly harm the benign application.%

Furthermore, setting up execution environment for \sw requires knowledge of low-level architecture and assembly instructions 
in order to deal directly with register state. Assuming such knowledge might be unrealistic or unnecessary for a typical application 
developer, who codes applications using a high-level programming language, e.g., \texttt{C}.
To this end, we provide an API to \sw that takes care of necessary configuration on the application's behalf.
\acron API implements appropriate configurations, saves the application state, calls \sw, and resumes the software 
state after \sw execution. This makes all \acron specifics transparent to the application developer, making \RA an 
easy-to-use service: a simple function call.

\acron API implementation is shown in Figure~\ref{fig:API}. It starts by copying \vrf's \chal to the designated physical memory 
location that will be read by \sw.
Next, to comply with atomicity, it disables interrupts. Before calling \sw, the current value of the stack pointer\footnote{In MSP430,
the stack pointer is the register $r1$.} is saved and set to point to the base of the secure stack. This is necessary to comply with the 
Key Confidentiality property which ensures that \sw can only execute on the reserved stack. With that last step, software state is 
ready for \sw execution. After \sw execution, the original stack pointer value and values of the registers used to store the original stack 
pointer are restored and interrupts are enabled.
At this point, execution of the application can continue and the application can reply to \vrf's request with the attestation result.

\noindent \emph{\textbf{Remark.} \acron API makes it easy for programmers to comply with \hw requirements before calling \sw.
The API itself does not (and should not) provide any security properties, since it is executed before and after \sw invocation  
to save and resume application execution state. Such code is not part of \sw and resides in regular program memory,
where it is treated accordingly by \hw's access control rules.
In summary, it is \textbf{not part of} the verified design.
}

\section{FPGA Deployment and Sample Application\label{comparison}}

\acron's design can be synthesized and deployed in real IoT/CPS environments.
To demonstrate its practicality and ease of use we provide, as part of \acron's repository~\cite{public-code},
a ready-to-go synthesize-able version of the architecture for the commodity FPGA Basys3~\footnote{https://store.digilentinc.com/basys-3-artix-7-fpga-trainer-board-recommended-for-introductory-users/} (depicted in Figure~\ref{fig:fpga_pic}).
The design can be easily ported to other FPGA models by mapping the input and output ports accordingly in the Verilog's constraints file.

\begin{figure}[ht]
\centering
\resizebox{0.65\columnwidth}{!}{ 
\includegraphics[width=0.8\columnwidth]{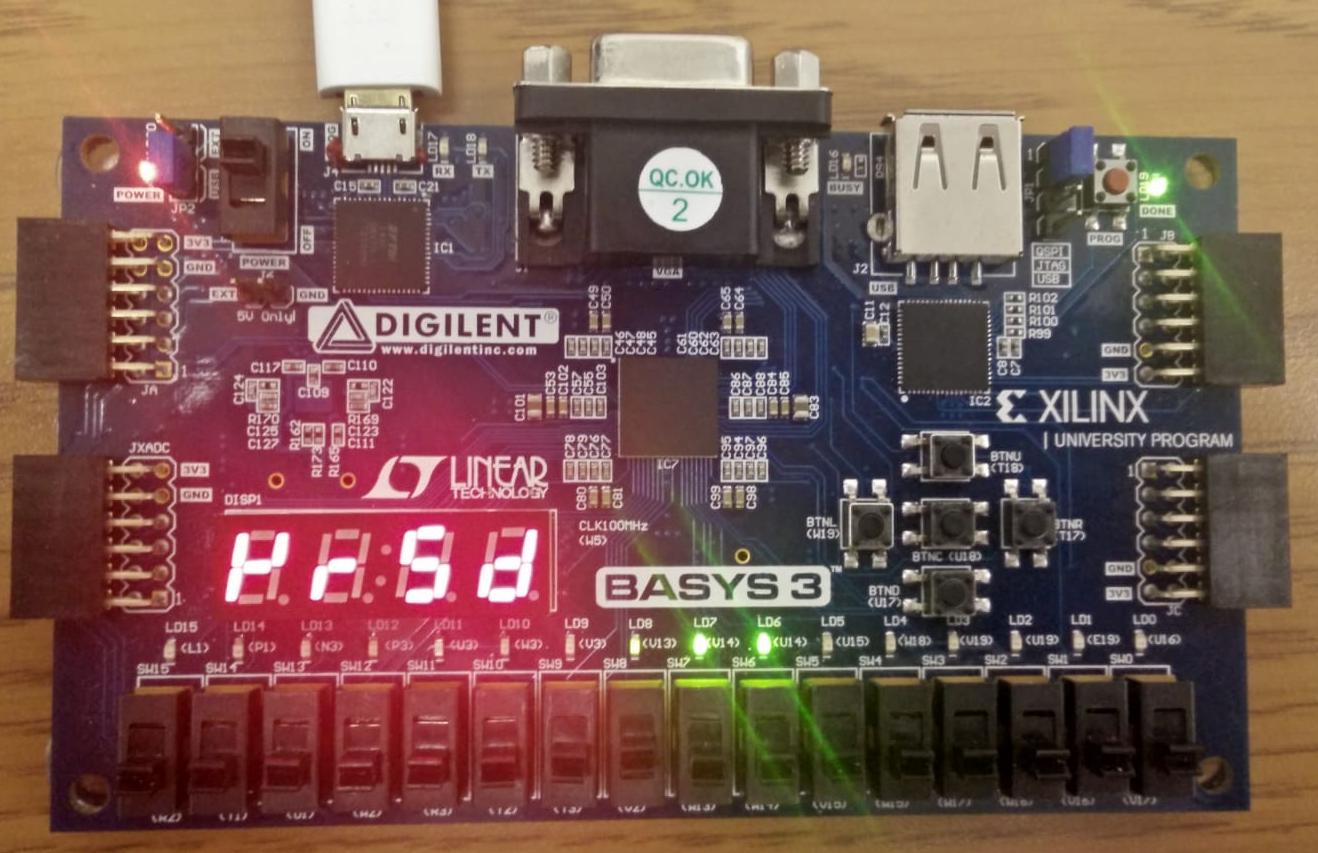}
}
	\caption{\small Basys3 FPGA running \acron's HW architecture depicted in Figure~\ref{fig:arch}}
	\label{fig:fpga_pic}
\end{figure}

\begin{figure}[!hbtp]
\begin{lstlisting}[basicstyle=\tiny, numberstyle=\tiny, xleftmargin=.12\textwidth, xrightmargin=.12\textwidth]
int main() {
  uint8_t challenge[32];
  uint8_t response[32];
  my_memset(challenge, 32, 1);

  P3DIR  =  0xFF;
  P3OUT  =  0x00;
  uint32_t count = 0; 
  volatile uint8_t buffer = 0;
  while (1) {
     while (count < 3000000) {
        count ++;    
     }
     count = 0;
     P3OUT++;
     if (P3OUT % 10 == 0) {
        buffer = P3OUT;
        P3OUT = 0xFF;
        VRASED(challenge, response);
        count = 0;
        P3OUT = buffer;
    }

  }
  return 0; 
}
\end{lstlisting}
\caption{Toy MSP430 application demo running \acron's \RA in real HW}\label{fig:toy_code}
\end{figure}

Figure~\ref{fig:toy_code} presents a toy sample application written in MSP430 \texttt{C}.
In it, $P3$ is an 8-bit General Purpose Input Output (GPIO) port which, in the synthesized HW, is connected to
LEDs 0-7 of Basys3 FPGA.
Lines 2-4 allocate buffers for the attestation challenge and response and initialize the challenge buffer.
In practice, the challenge is received from \vrf via communication channels such as MSP430 Universal Asynchronous Receiver/Transmitter (UART).
For the sake of clarity and brevity we omit the communication step from the example and set the challenge to a constant.
Line 6 in Figure~\ref{fig:toy_code} sets P3 GPIO as output (i.e., an actuator port) and line 7 initializes all 8 bits to zero, making all LEDs initially off.
The main application loop starts at line 10;
at every iteration an artificial delay of 3 million integer increments is introduced and then P3 output value is incremented.
This results in a binary counter being displayed on the 8 LEDs.
At every time the counter value reached a multiple of 10 (line 16), all LEDs turn on (line 18) and the \RA procedure is called
in line 19 (by default VRASED \RA is computed in the entire program memory, but the range is configurable and allows for data memory attestation as well).
The LEDs remain on until the end of \RA computation.
After completion of attestation, the attestation result is saved in the \texttt{response} buffer and the counter resumes.
In practical applications, such result can be reported back to \vrf (via UART) as an authenticated measurement of the device's state.
A demo video of this application running on real hardware and computing \RA in a small fraction of a second is available on \acron's repository~\cite{public-code}.

\end{document}